\documentclass[11pt,a4paper]{article}
\usepackage[colorlinks=true]{hyperref}
\hypersetup{urlcolor=blue, citecolor=red, linkcolor=blue}
\usepackage[table]{xcolor}
\usepackage{color}
\usepackage[all]{xy}
\usepackage{amsfonts}
\usepackage{amssymb}
\usepackage{amsmath}
\usepackage[utf8]{inputenc}
\usepackage{amsthm}
\usepackage[active]{srcltx}
\usepackage{hyperref} 

\voffset = - 10mm

\newtheorem{theorem}{Theorem}[section]
\newtheorem{corollary}[theorem]{Corollary}
\newtheorem{lemma}[theorem]{Lemma}
\newtheorem{proposition}[theorem]{Proposition}
\theoremstyle{definition}
\newtheorem{definition}[theorem]{Definition}
\newtheorem{remark}[theorem]{Remark}
\newtheorem{example}[theorem]{Example}
\newcommand\xqed[1]{%
  \leavevmode\unskip\penalty9999 \hbox{}\nobreak\hfill
  \quad\hbox{#1}}
\newcommand\demo{\xqed{$\triangle$}}
\title{\sf Multisymplectic structures and invariant tensors for Lie systems}
\author{\sf X. Gr\`acia$^a$, J. de Lucas$^b$, M.C. Mu\~{n}oz-Lecanda$^a$, and S. Vilari\~no$^c$}

\date{%
$^a$\textit{Dept.\ of Mathematics, Universitat Polit\`ecnica de Catalunya, Barcelona, Spain}\\
$^b$\textit{Katedra Metod Matematycznych Fizyki, Uniwersytet Warszawski, Warszawa, Poland}\\
$^c$\textit{Centro Universitario de la Defensa \& IUMA, Zaragoza, Spain}\\[3ex]
\textsf{19 August 2018}
}

\numberwithin{equation}{section}
\begin{document}
\maketitle

\begin{abstract}
A {\it Lie system} is the non-autonomous system of differential equations describing the integral curves of a non-autonomous vector field 
taking values in a finite-dimensional Lie algebra of vector fields, a so-called {\it Vessiot--Guldberg Lie algebra}.
This work pioneers the analysis of Lie systems admitting a Vessiot--Guldberg Lie algebra of Hamiltonian vector fields relative to a multisymplectic structure: 
the {\it multisymplectic Lie systems}. 
Geometric methods are developed to consider a Lie system as a multisymplectic one. 
By attaching a multisymplectic Lie system via its multisymplectic structure with a tensor coalgebra, 
we find methods to derive superposition rules, constants of motion, and invariant tensor fields relative to the evolution of the multisymplectic Lie system. 
Our results are illustrated with examples occurring in physics, mathematics, and control theory. 
\end{abstract}


%


\section{Introduction}
\label{section:intro}

A \textit{Lie system} is a non-autonomous first-order system of ordinary differential equations in normal form 
whose general solution can be expressed as an autonomous function, 
a so-called {\it superposition rule}  
\cite{CGM00,Dissertationes,LS,PW}, 
depending on a generic finite family of particular solutions and a set of constants. 
Standard examples of Lie systems are most types of Riccati equations 
\cite{GL13,PW} 
and non-autonomous first-order affine systems of ordinary differential equations 
\cite{Dissertationes}.

The Lie--Scheffers theorem \cite{CGM07,LS,OG00} states that a Lie system amounts to a $t$-dependent vector field taking values in a 
\textit{Vessiot--Guldberg Lie algebra} 
\cite{Dissertationes,Ib09}.
The later property gave rise to a number of methods for the determination of superposition rules 
based upon the integration of systems of ordinary and/or partial differential equations 
that are simpler to solve than the Lie systems they describe 
\cite{CGL18,CGM07,Dissertationes,PW}. 
At the same time, the Lie--Scheffers theorem also showed that being a Lie system is rather the exception than the rule \cite{Dissertationes}. 
Despite this, Lie systems admit relevant physical and mathematical applications, as witnessed by the many works on the topic \cite{BHLS15,CCJL18,CGM00,Ru10,GMR97,Ib09,Ra06,Te60,PW}.

Recently, a lot of attention has been paid to Lie systems admitting a Vessiot--Guldberg Lie algebra of Hamiltonian vector fields and/or Lie symmetries 
relative to a geometric structure: 
Poisson and symplectic 
\cite{BBHLS13,BCHLS13,BHLS15,CCJL18,CGM00,CLS13,Ru10}, 
Dirac 
\cite{CCJL18,CGLS14}, 
$k$-symplectic 
\cite{LV15}, 
Jacobi 
\cite{HLS15}, 
Riemann 
\cite{HLT17,LL18}, 
and others 
\cite{CCJL18,LL18}. 
We say that these Lie systems admit 
\emph{compatible geometric structures}. 
Although such Lie systems represent a relatively small subclass of all Lie systems \cite{BBHLS13,GL17,GMR97,LL18}, 
they seem to admit more applications than 
Lie systems without compatible geometric structures 
\cite{BHLS15,LL18}.

Geometric structures compatible with Lie systems allow for the algebraic construction of superposition rules, 
constants of motion, 
and other evolution invariants of Lie systems 
without solving systems of partial and/or ordinary differential equations \cite{BCHLS13,CGLS14,LV15}
as in standard methods 
\cite{CGM07,Dissertationes,PW}. 
Geometric structures also explain the geometric properties of superposition rules 
\cite{BCHLS13,BHLS15}, 
and lead to the investigation of non-Lie systems 
\cite{BCFHL18} 
as well as physical and mathematical problems 
\cite{BHLS15,CCJL18,GMR97,LL18}.

The first aim of this paper is to introduce Lie systems admitting a Vessiot--Guldberg Lie algebra of Hamiltonian vector fields relative to a multisymplectic structure \cite{CIL-96b}. 
Such systems are called 
\emph{multisymplectic Lie systems} 
and the associated multisymplectic form is referred to as a 
\emph{compatible multisymplectic form}. 
To motivate the relevance of multisymplectic Lie systems, 
we provide examples appearing in physical and mathematical problems, 
e.g.\ in the study of the Schwarz derivative  \cite{Be07,LG99,OT09}, 
control systems 
\cite{Ni00,Ra06}, 
and diffusion equations 
\cite{SSV14}.

The theory of multisymplectic Lie systems plays a more relevant role than Lie systems related to other geometric structures 
\cite{BCHLS13,CGLS14,CLS13,LL18,LV15}. 
For instance, since symplectic structures are multisymplectic ones, 
the Lie--Hamilton systems related to symplectic structures 
\cite{CLS13} 
can be considered as multisymplectic Lie systems. 
Moreover, we prove that the hereafter called  \emph{locally automorphic Lie systems}, 
which are Lie systems locally diffeomorphic to the very relevant \emph{automorphic Lie systems} 
\cite{BM09,CGM00,KCD93}, 
can always be studied through multisymplectic Lie systems 
(cf.\ Theorem \ref{th:invariant_form}). Consequently, multisymplectic Lie systems admit much more applications 
than other types of Lie systems with compatible geometric structures which cannot be used to study all above-mentioned types of Lie systems, e.g. Lie--Hamilton ones  
(cf.\ \cite{CGLS14,LV15}).
As a byproduct of our techniques, 
we find that multisymplectic Lie systems can be frequently endowed with other compatible geometric structures, 
e.g.\ Dirac and $k$-symplectic structures 
\cite{CGLS14,LV15}, 
which can be used to apply previously known techniques, 
for instance, 
to obtain their constants of motion or superposition rules.

More specifically, 
Theorems \ref{th:invariant_form}, \ref{th:MTS}, 
and Corollary \ref{cor:InvVolumeForm}, 
give methods to endow locally automorphic Lie systems 
with a compatible multisymplectic structure 
and other invariants, 
e.g.\ compatible presymplectic structures which can also be understood as Dirac structures 
\cite{CGLS14}.  
As the local diffeomorphisms mapping 
locally automorphic Lie systems into automorphic ones 
are difficult to be obtained explicitly and they are generally locally defined, 
it is unlikely that automorphic Lie systems can be used to study directly
locally automorphic Lie systems.
Despite that, the existence of the local diffeomorphism 
(see Theorem \ref{Trivial}) 
is at the core of all methods in Section \ref{section:LS-IF} 
to infer the existence of invariant geometric structures, 
e.g.\ multisymplectic volume forms, 
for locally automorphic Lie systems.  
Such compatible multisymplectic forms 
can be obtained algebraically in an easy manner under mild conditions, 
e.g.\ when the locally automorphic Lie system 
admits a unimodular Vessiot--Guldberg Lie algebra 
(see Corollary \ref{cor:InvVolumeForm}). 
Additionally, other accessory results concerning the properties of locally automorphic Lie systems are detailed in 
Corollaries 
\ref{cor:superposition_rule_locally_automorphic_LS} and 
\ref{cor:superposition_rule_locally_automorphic_LS_2}. 

Next, compatible multisymplectic forms are employed to study superposition rules and constants of motion for multisymplectic Lie systems in geometric and algebraic terms. 
Our methods do not require the integration of systems of partial or ordinary differential equations as in most standard methods in the literature 
\cite{CGM07,PW}. 
Our procedures also avoid the transformation of a system of differential equations onto a normal form, as employed in several works 
\cite{GL17,SW84}. 
As a byproduct, 
our approach also retrieves algebraically and geometrically 
invariants and geometric structures related to Lie systems appearing in previous works  
\cite{BCHLS13,BHLS15,CGLS14,LV15}. 
These structures were obtained in the above-mentioned articles in an {\it ad-hoc} manner 
or by solving systems of PDEs. 
Then, our work simplifies their derivation. 

Remarkably, we show that the coalgebra method to derive superposition rules for Lie--Hamilton systems developed in 
\cite{BCHLS13,BHLS15} 
can be retrieved as a particular case of our techniques 
when it concerns Lie--Hamilton systems related to symplectic forms. 
Moreover, 
our methods give rise to tensor field invariants for multisymplectic Lie systems from invariants of tensor algebras, 
which are more general than the invariant structures appearing in the standard coalgebra method, e.g. Casimir elements and invariant functions \cite{BCHLS13}. 

More specifically, 
a multisymplectic Lie system $(N,\Theta,X)$, 
where $X$ is a Lie system on a manifold $N$ with a compatible multisymplectic form $\Theta$, 
is endowed with a finite-dimensional Lie algebra $\mathfrak{M}$ 
of Hamiltonian forms for one of its Vessiot--Guldberg Lie algebras: 
a so-called {\it Lie--Hamilton algebra} of~$X$. 
If $\mathfrak{g}$ is an abstract Lie algebra isomorphic to~$\mathfrak{M}$, 
then the adjoint representation of $\mathfrak{g}$ can be extended to a Lie algebra representation on the tensor algebra, $T(\mathfrak{g})$, 
which makes the latter into a 
\emph{$\mathfrak{g}$-module} 
\cite{Va84}. 
Similarly, the symmetric and Grassmann algebras, 
$S(\mathfrak{g})$ and $\mathsf{\Lambda}(\mathfrak{g})$, 
can be considered as 
$\mathfrak{g}$-submodules of $T(\mathfrak{g})$. 
Moreover, we endow 
$T(\mathfrak{g})$, 
$S(\mathfrak{g})$, 
$\mathsf{\Lambda}(\mathfrak{g})$ 
with {\it coalgebra structures} 
(see \cite{BCHLS13,Pressley} for details), 
which are extended to the tensor products 
$T^{(m)}(\mathfrak{g})=
T(\mathfrak{g})\boxtimes\stackrel{m}{\ldots} \boxtimes T(\mathfrak{g})$, 
$S^{(m)}(\mathfrak{g})=
S(\mathfrak{g})\boxtimes\stackrel{m}{\ldots} \boxtimes S(\mathfrak{g})$, and 
$\mathsf{\Lambda}^{(m)}(\mathfrak{g})=
\mathsf{\Lambda}(\mathfrak{g})\boxtimes\stackrel{m}{\ldots} \boxtimes \mathsf{\Lambda}(\mathfrak{g})$. 
Previous structures are then represented as covariant tensor fields on~$N$ and $N^m$ in such a way that 
the $\mathfrak{g}$-invariants in $T(\mathfrak{g})$ 
(or its  $\mathfrak{g}$-submodules 
$\mathsf{\Lambda}(\mathfrak{g})$ and $S(\mathfrak{g})$) 
give rise to tensor invariants for~$X$ and its diagonal prolongations 
\cite{CGM07}; 
see diagrams (\ref{StarTrekki1}) and (\ref{StarTrekki2}) for details. 
This is employed to obtain constants of motion and superposition rules for~$X$ 
\cite{Dissertationes}.

Our approach shows that invariants and superposition rules for multisymplectic Lie systems can be obtained through Casimir elements of universal enveloping algebras, 
which can be understood as symmetric tensors in $T(\mathfrak{g})$, 
or co-cycles of the Chevalley--Eilenberg cohomology of $\mathfrak{g}$ (see \cite{Va84}), 
which are understood as antisymmetric tensors of $T(\mathfrak{g})$. 
Moreover, this method gives rise to obtaining $k$-symplectic or presymplectic structures compatible with Lie systems, 
which allows for the application of the techniques in \cite{CGLS14,LV15} to study multisymplectic Lie systems.

As an application, our methods are employed to study superposition rules for multisymplectic Lie systems related to locally automorphic Lie systems. 
In particular, the cases of Schwarz equations and Riccati-type diffusion systems are studied in detail, 
while control and Darboux--Brioschi--Halphen systems are used to illustrate some results and/or techniques 
\cite{Halphen,Ni00,Ra06}.

The structure of the paper goes as follows. 
Section~2 surveys several fundamental concepts on Lie systems and multisymplectic structures to be used hereupon. 
Section~3 is devoted to motivating the definition of multisymplectic Lie systems and to illustrating some of its applications in the physics and mathematics literature. 
Methods for the calculation of compatible multisymplectic forms for locally automorphic Lie systems are described in Section~4. 
The use of multisymplectic structures and tensor coalgebras for the determination of invariants, constants of motion, and superposition rules for multisymplectic Lie systems is developed in Section~5. 
Section~6 summarises the results of the work and provides some hints on future research. 
Additionally, the Appendix contains the proof of some technical results.

\section{Some basic concepts and notations}
\label{section:basic_concepts}

Unless otherwise stated, we assume all mathematical objects to be real, smooth, and globally defined.
This permits us to omit minor technical problems so as to highlight the main aspects of our theory. 
$N$ will hereafter represent an $n$-dimensional connected manifold. All remaining manifolds are considered to be finite-dimensional and connected if not stated otherwise.

\subsection{Generalised distributions and t-dependent vector fields}
\label{subsection:generalised_distributions}

Let $V$ be a Lie algebra.
Given two subsets $\mathcal{A}, \mathcal{B} \subset V$, 
we write $[\mathcal{A},\mathcal{B}]$ for the linear space spanned by the Lie brackets between elements of $\mathcal{A}$ and $\mathcal{B}$.
Meanwhile, ${\rm Lie}(\mathcal{B})$ stands for the smallest Lie subalgebra of $V$ containing $\mathcal{B}$.

Given a vector bundle $\rho \colon P\rightarrow N$, we denote by $\Gamma(\rho)$
its $C^\infty(N)$--module of sections.
In particular, if $\tau_N\colon {\rm T}N\rightarrow N$ is the tangent bundle projection, then
$\mathfrak{X}(N) = \Gamma(\tau_N)$
designates the $C^\infty(N)$--module of vector fields on~$N$.

Remember that a
{\it generalised distribution}
$\mathcal{D}$ on a
manifold $N$ is a function mapping each $x\in N$ to a linear
subspace $\mathcal{D}_x\subset {\rm T}_xN$.
A vector field $Y$ on~$N$ is said to
take values in $\mathcal{D}$, in short $Y\in\mathcal{D}$, when
$Y_x\in\mathcal{D}_x$ for all $x\in N$. 

An arbitrary set $\mathcal{V}$ of vector fields on~$N$ generates a generalised distribution $\mathcal{D}^\mathcal{V}$ on $N$ by considering, at each point $x\in N$, the linear span of all of its vector fields:
$\mathcal{D}^\mathcal{V}_x = \mathrm{span}\{X_x \mid X \in \mathcal{V} \}$.
As these vector fields are smooth by assumption, the generalised distribution $\mathcal{D}^\mathcal{V}$ is smooth \cite{DLPR12}. 
Along the paper all distributions are assumed to be smooth.

The dimension of $\mathcal{D}_x$ is called the {\it rank} of $\mathcal{D}$ at~$x$.
A generalised distribution $\mathcal{D}$ is
{\it regular at} $x'\in N$ when,
in a neighbourhood of~$x'$,
the distribution has constant rank. 
The generalised distribution
$\mathcal{D}$ is called \emph{regular} when its rank is constant on the whole~$N$.

A {\it $t$-dependent vector field} on $N$ is a map
$X \colon (t,x) \in \mathbb{R} \times N  \mapsto  X(t,x) \in {\rm T}N$
such that $\tau_N\circ X=\pi_2$,
where
$\pi_2 \colon (t,x) \in \mathbb{R} \times N  \mapsto  x \in N$. An {\it integral curve} of $X$ 
is a curve $\gamma \colon \mathbb{R} \to N$
such that 
\begin{equation}\label{Eq:Sys}
\frac{{\rm d}\gamma}{{\rm d} t}(t) = X(t,\gamma(t))
\,,
\quad 
\forall t\in \mathbb{R}
\,.
\end{equation}
Then,
$\tilde\gamma(t) = (t,\gamma(t))$
is an integral curve
of the {\it suspension} $\tilde X$ of~$X$,
namely the vector field 
$\displaystyle
\tilde X = \partial/\partial t + X$ 
on the product manifold 
$\mathbb{R}\times N$
\cite{FM}.
Conversely,
if $\tilde\gamma \colon \mathbb{R} \to \mathbb{R} \times N$ 
is an integral curve of the suspension~$\tilde X$
satisfying
$(\pi_1 \circ \tilde\gamma)(t) = t$
for every~$t$,
where
$\pi_1 \colon (t,x) \in \mathbb{R} \times N  \mapsto  t \in \mathbb{R}$,
then
$\gamma = \pi_2 \circ \tilde\gamma$
is an integral curve of~$X$.

Every $t$-dependent vector field $X$ gives rise to a unique system (\ref{Eq:Sys}) describing its integral curves. 
Also, every system (\ref{Eq:Sys}) describes the integral curves $\bar\gamma:t\in \mathbb{R}\rightarrow (t,\gamma(t))\in \mathbb{R}\times N$ of the suspension of a unique $t$-dependent vector field $X$. 
This motivates to use $X$ to designate both a $t$-dependent vector field and its associated system (\ref{Eq:Sys}), indistinctly. 

Notice that giving a $t$-dependent vector field~$X$ 
amounts to giving 
a family of vector fields
$\{X_t\}_{t\in\mathbb{R}}$ on~$N$,
with 
$X_t \colon x \in N  \mapsto  X(t,x) \in {\rm T}N$
\cite{Dissertationes}. This enables us to relate $t$-dependent vector fields to several geometric structures given in the following definition.

\begin{definition}
\label{def:smallest_Lie_algebra}
Let $X$ be a $t$-dependent vector field on~$N$. 
The {\it smallest Lie algebra} of~$X$
is the smallest real Lie algebra, $V^X$, 
containing the vector fields $\{X_t\}_{t\in\mathbb{R}}$, namely
$V^X={\rm Lie}(\{X_t\}_{t\in\mathbb{R}})$. 
The {\it associated distribution} of~$X$ is the generalised distribution on~$N$ spanned by the vector
fields of the smallest Lie algebra $V^X$, that is,
$\mathcal{D}^{V^X}$. 
\end{definition}

It can be proved that
the rank of $\mathcal{D}^{V^X}$
must only be
constant on the connected components of an open and dense subset of~$N$,
where the distribution becomes regular, involutive, and integrable (see \cite{CLS13}).
The most relevant instance for us is when $\mathcal{D}^{V^X}$ is determined by a finite-dimensional $V^X$ 
and hence the distribution becomes integrable on the whole~$N$ in the sense of Stefan--Sussmann 
\cite[p.\,63]{JPOT}.
Among other reasons, the associated distribution is important to study superposition rules for Lie systems \cite{Dissertationes}.

\subsection{Lie systems}\protect\label{subsection:Lie systems}

Let us now turn to some fundamental notions appearing in the theory of Lie systems (see \cite{Dissertationes} for details).

\begin{definition}\label{def:superposition_rule} A {\it superposition rule} depending on $m$ particular solutions for a system~$X$ in $N$ 
is a function
$\Phi \colon N^{m} \times N \rightarrow N$,
$x=\Phi(x_{(1)}, \ldots,x_{(m)};\lambda)$,
such that the general solution, $x(t)$, of $X$ can be brought into the form
$x(t)=\Phi(x_{(1)}(t), \ldots,x_{(m)}(t);\lambda),$
where $x_{(1)}(t),\ldots,x_{(m)}(t)$ is any generic family of
particular solutions and $\lambda$ is a point of $N$ to be related to initial conditions. 
A {\it Lie system} is a system of first-order ordinary differential equations which admits a superposition rule.
\end{definition}

The conditions ensuring that a $t$-dependent system possesses a superposition rule are
stated
by the {\it Lie--Scheffers theorem} \cite{CGM07,LS}.

\begin{theorem}
\label{th:characterization_superposition_rule}
A $t$-dependent vector field $X$ admits a superposition rule if and only if
$X$ can be written as
$X = {{\sum_{\alpha=1}^r}} b_\alpha(t) X_\alpha$,
for a certain family $b_1(t),\ldots,b_r(t)$ of functions and a
collection $X_1,\ldots,X_r$ of vector fields spanning
an $r$-dimensional real Lie algebra.
\end{theorem}

In other words,
$X$ admits a superposition rule if and only if
its smallest Lie algebra $V^X$ is finite-dimensional. 
Normally, a Lie system is given in the form 
$X=\sum_{\alpha=1}^rb_\alpha(t)X_\alpha$, 
where the vector fields $X_\alpha$, with $\alpha=1,\ldots,r$, 
span a Lie algebra $V$ that may strictly contain $V^X$. 
Then, $V$ is called a
{\it Vessiot--Guldberg Lie algebra} of~$X$.

In view of the preceding theorem and comments,
from now on
we will denote a \emph{Lie system} as a triple 
$(N,X,V)$,
where $N$ is a manifold,
$X$ is a $t$-dependent vector field on~$N$ as given by Lie--Scheffers theorem,
and $V$ is a Vessiot--Guldberg Lie algebra of~$X$.

Now we will see how the integration of a Lie system on a manifold 
can be reduced to 
the integration of a Lie system on a Lie group \cite{CGM00,Ve1904}.
To this end, we need to recall some facts about Lie group actions.

Consider a (left) Lie group action
$\varphi \colon G \times N \to N$,
or more generally
a local Lie group action
$\varphi \colon D \subset G \times N \to N$,
and denote by
$\varphi_{x} = \varphi(\cdot,x)$ 
the partial map defined, in general, on an open neighbourhood of the neutral element of~$G$ (see \cite{Palais} for details).
Then one defines, 
for every $\xi \in \mathrm{T}_eG$,
the fundamental vector field $\xi_N \in \mathfrak{X}(N)$,
given by
$\xi_N(x) = T_e \varphi_{x} (\xi)$.
The tangent space $\mathrm{T}_eG$ is in bijection with 
the left-invariant, $\mathfrak{X}_L(G)$,
and right-invariant, $\mathfrak{X}_R(G)$, 
vector fields of~$G$;
let us denote by $\xi^L$ and $\xi^R$ 
the respective invariant vector fields associated with 
$\xi \in \mathrm{T}_eG$.
By convention,
$\mathfrak{g} = \mathrm{T}_eG$ 
inherits its Lie algebra structure from $\mathfrak{X}_L(G)$.
Then
(cf.\ \cite[Ch.\,20]{Lee})
\begin{enumerate}
\itemsep 0pt
\item 
The map 
$\hat\varphi \colon \mathfrak{g} \to \mathfrak{X}(N)$,
$\xi \mapsto \xi_N$,
is a Lie algebra antihomomorphism
(the \emph{infinitesimal generator} of~$\varphi$).
\item
For every $x \in N$,
the \emph{right}-invariant vector field $\xi^R$
and the fundamental vector field $\xi_N$ 
are $\varphi_x$-related.
\end{enumerate}

In general, an antihomomorphism 
$\mathfrak{g} \to \mathfrak{X}(N)$ 
is called a (left) Lie algebra action, and,
when $\mathfrak{g}$ is finite-dimensional,
it can be integrated to a Lie group action.
More precisely:

\begin{theorem}
Let $\mathfrak{g}$ be a finite-dimensional Lie algebra,
and let $G$ be a Lie group with Lie algebra $\mathfrak{g}$.
Given a Lie algebra action $\hat\varphi$ of $\mathfrak{g}$ on a manifold~$N$,
there exists a local Lie group action 
$\varphi$ of~$G$ on~$N$
such that $\hat\varphi$ is the infinitesimal generator of~$\varphi$.
If $G$ is simply connected and the vector fields of the image of~$\hat\varphi$ are complete,
then $\varphi$ can be supposed to be a global Lie group action. 
\end{theorem}
The proof of these results, and other related facts, can be found, for instance, in
\cite[p.\,58]{Palais}
\cite[p.\,529]{Lee}
\cite[p.\,207]{Bo68}.
This Lie group action is the device that relates the Lie system on~$N$ to a Lie system on~$G$:

\begin{theorem}
\label{Th:XG-X}
Let $(N,X,V)$ be a Lie system of the form
$X = {{\sum_{\alpha=1}^r}} b_\alpha(t) X_\alpha$,
where
$X_1,\ldots,X_r$ is a basis of the Vessiot--Guldberg
Lie algebra~$V$.
Let $G$ be a Lie group whose Lie algebra is isomorphic to $V$,
and let $\varphi$ be a local group action of $G$ on~$N$
as given by the preceding theorem.
Let $X_\alpha^R$ be the right-invariant vector fields on~$G$
related to the vector fields $X_\alpha$ through the action $\varphi$.
Then
\begin{enumerate}
\item 
The triple 
$(G,X^G,V^G)$,
where
\begin{equation}
\label{Eq:EquLie}
X^G(t,g) = \sum_{\alpha=1}^r b_\alpha(t) X_\alpha^R(g)
\end{equation}
and
$V^G = \mathfrak{X}_R(G)$,
is a Lie system on~$G$.
\item
For every $x_0 \in N$ and $t\in \mathbb{R}$,
the vector field
$X_t^G$ is $\varphi_{x_0}$-related with~$X_t$; 
namely, 
the $t$-dependent vector fields $X^G$ and $X$ are $\varphi_{x_0}$-related.
\item
If $g(t)$ is the integral curve of $X^G$ with $g(0)=e$,
and $x_0 \in N$,
then 
$x(t) = \varphi(g(t),x_0)$
is the integral curve of~$X$ with $x(0)=x_0$.
\end{enumerate}
\end{theorem}
The proof of this result is almost immediate:
the vector fields $X_\alpha^R$ span the Lie algebra
$\mathfrak{X}_R(G)$, the $t$-dependent vector field
$X^G$ is $\varphi_{x_0}$-related with~$X$,
and this implies that 
$\varphi_{x_0}$ maps integral curves of $X^G$ to integral curves of~$X$ (see
\cite{CGM00,Dissertationes}
for details).

In this manner, if $\varphi$ is explicitly known, 
then finding all the integral curves of~$X$
reduces to finding one particular integral curve of~(\ref{Eq:EquLie}).
Conversely, the general solution of $X$ enables us to construct 
the integral curve of (\ref{Eq:EquLie}) with $g(0)=e$
by solving an algebraic system of equations obtained through~$\varphi$;
this can be used to solve other Lie systems
\cite{Dissertationes}.

Lie systems of the form (\ref{Eq:EquLie}) are sometimes called 
\emph{automorphic Lie systems} in the literature 
\cite{BS08,Ve1904}. 
Due to their specific structure, these systems admit invariant forms relative to their evolution;
this will be studied later in Section~4. 
We will say that the automorphic Lie system 
$(G,X^G,V^R)$,
where $V^R = \mathfrak{X}_R(G)$, 
is \emph{associated} with $(N,X,V)$.

It is worth noting that the automorphic Lie system related to the $t$-dependent right-invariant vector field (\ref{Eq:EquLie}) 
is invariant under the right multiplications 
$R_g: g'\in G \rightarrow g'g\in G$. 
Then, if $g_1(t)$ is a particular solution to the system $X^G$ with initial condition $g_1(0)$, 
then $R_gg_1(t)$ is a new particular solution of the same system with initial condition $g_1(0)g$. 
Hence, the general solution to this system can be written as 
$g(t)=\Phi^G(g_1(t);g) := R_g \,g_1(t)$, 
which gives rise to the superposition rule 
$\Phi^G \colon (g_1;g) \in G\times G \mapsto R_gg_1 \in G$ 
depending on a sole particular solution 
(see \cite{Dissertationes} for details).

\subsection{Multisymplectic manifolds}\label{subsection:multisymplectic}

This section addresses the main properties of multisymplectic structures
(see \cite{CIL-96a,CIL-96b} for details). 
We hereafter write $\Omega(N)$ and $\Omega^k(N)$ for the spaces of differential forms and differential $k$-forms on~$N$, respectively.

A differential $k$-form $\omega$ on~$N$
is called \emph{1-nondegenerate}
if, for every $p \in N$,
the inner contraction
$\iota_{X_p} \omega_p = 0$
if and only if
$X_p = 0$.
In other words,
$\omega$ is 1-nondegenerate if and only if the vector bundle morphism
\[
\begin{array}{rrcl}
\omega^\flat \colon & \mathrm{T}N & \to & \mathsf{\Lambda}^{k-1} \mathrm{T}^*N
\\
& X_p & \mapsto & \iota_{X_p} \omega_p
\end{array}
\]
is injective.
In this case,
the corresponding morphism of $\mathcal{C}^\infty(N)$-modules
$\hat\omega\colon\mathfrak{X}(N) \to \Omega^{k-1}(N)$,
$X  \mapsto  \iota_X \omega$,
is also injective.

\begin{definition}\label{def:Mult}
A \emph{multisymplectic $k$-form} on $N$ is a closed and $1$-nondegenerate differential 
$k$-form $\Theta \in \Omega^k(N)$.
A \emph{multisymplectic manifold} of degree $k$ is a pair $(N,\Theta)$,
where $\Theta$ is a multisymplectic $k$-form on $N$.
\end{definition}

Thus, multisymplectic 2-forms are just symplectic forms.
Multisymplectic $n$-forms on $N$ coincide with volume forms.
In what follows, we assume that $\dim N\geq 2$. 
In this case,
every multisymplectic $k$-form has degree $k \geq 2$. 

\begin{definition}\label{def:multi_hamvf}
Let $(N,\Theta)$ be a multisymplectic manifold of degree~$k$.
A vector field $X$ on $N$ is \textit{locally Hamiltonian} if $\iota_X\Theta$ is closed;
this amounts to saying that $\Theta$ is invariant by~$X$, that is,
the Lie derivative of $\Theta$ relative to $X$ vanishes,
\[
\mathcal{L}_X \Theta = 0
\,.
\]
A vector field $X$ is \textit{globally Hamiltonian} if $\iota_X\Theta$ is exact;
that is, there exists a  differential $(k-2)$-form  $\Upsilon_X$ on $N$ such that
\[
\iota_X \Theta= d\Upsilon_X
\,.
\]
In this case, $\Upsilon_X$ is called a \textit{Hamiltonian form associated with~$X$}.
\end{definition}

For locally Hamiltonian vector fields $X, Y$,
we have
$\iota_{[Y,X]}\Theta =
d \iota_Y\iota_X\Theta$,
that is, their Lie bracket is a Hamiltonian vector field.
Therefore, the space of (locally) Hamiltonian vector fields of
$(N,\Theta)$ 
is a Lie algebra.

\begin{definition}\label{def:AH}
Let $(N,\Theta)$ be a multisymplectic manifold of degree~$k$.
\begin{itemize}
\item
Let 
$\xi,\zeta \in {\rm Im}\,\hat\Theta \subset \Omega^{k-1}(N)$, 
and let $X,Y\in\mathfrak{X}(N)$ be the unique vector fields such that $\iota_{X}\Theta=\xi$ and $\iota_{Y}\Theta=\zeta$.
The {\sl bracket between $\xi$ and $\zeta$} is defined by
\begin{equation}
\{\xi,\zeta\} = \iota_{[Y,X]}\Theta\ \in {\rm Im}\,\hat\Theta
\,.
\label{braform}
\end{equation}
It is immediate from its definition that this bracket satisfies the Jacobi identity and becomes a Lie bracket.

\item
The {\sl bracket between Hamiltonian forms}
is defined in the following way:
let
$\Upsilon_X,\Upsilon_Y \in \Omega^{k-2}(N)$ 
be Hamiltonian forms, corresponding to the Hamiltonian vector fields
$X,Y\in\mathfrak{X}(N)$. 
Then, we define
$$
\{\Upsilon_X,\Upsilon_Y\} = \iota_Y\iota_X\Theta \,.
$$
\end{itemize}
\end{definition}
It can be proved that the bracket of differential $(k-2)$-forms needs not be a Lie bracket for $k>2$ \cite{CIL-96a}.

Although the brackets for $k-1$ and $k-2$ differential forms have been denoted in the same way, 
this will not lead to any confusion and it will simplify the notation. 

From the above definitions and the properties of the Lie bracket we have that
\begin{equation}
d\{\Upsilon_X,\Upsilon_Y\}=
d\iota_Y\iota_X\Theta=
\iota_{[Y,X]}\Theta=
\{d\Upsilon_X,d\Upsilon_Y\} \,.
\label{relation1}
\end{equation}
As a consequence of the equality
$d\{\Upsilon_X,\Upsilon_Y\}=\iota_{[Y,X]}\Theta$, 
we have that $[Y,X]$ is a Hamiltonian vector field
which has $\{\Upsilon_X,\Upsilon_Y\}$ as a Hamiltonian form. 
Therefore, the space of (locally) Hamiltonian vector fields (${\rm Ham}_{\mathrm{loc}}(N)$) ${\rm Ham}(N)$ 
is a Lie algebra, 
and the function mapping a (locally) Hamiltonian vector field $X$ to $\iota_{X}\Theta$ is an injective Lie algebra anti-homomorphism. 

\medskip

Finally we will recall the notion of multivector field, 
which will be used to find constants of motion of multisymplectic Lie systems
(see \cite{EMR,FR} for more details).
An \textit{$\ell$-multivector field} on $N$ is a section of $\mathsf{\Lambda}^\ell(\mathrm{T}N)$.
An $\ell$-multivector field $Y$ is said to be \emph{decomposable} if there is a family of vector fields 
$Y_1,\ldots, Y_\ell \in \mathfrak{X}(N)$ 
such that 
$Y = Y_1 \wedge \ldots \wedge Y_\ell$.

Let $(N,\Theta)$ be a multisymplectic manifold of degree~$k$.
Generalising the notion of Hamiltonian vector field,
we say that an $\ell$-multivector field $Y$ is \textit{Hamiltonian} 
(with respect to~$\Theta$)
if there exists a $(k-\ell-1)$-form $\theta$ such that 
$\iota_Y \Theta = d\theta$. 
Additionally, $Y$ is \textit{locally Hamiltonian or multisymplectic} if 
$\mathcal{L}_Y \Theta = 0$ (see \cite{CIL-96a,CIL-96b}).

\subsection{Unimodular Lie algebras}
\label{subsection:unimodular}

This section surveys the notions of unimodular Lie algebras and unimodular Lie groups. These two definitions are necessary in the following parts of the paper, when the existence of multisymplectic structures compatible with Lie systems is addressed.

Consider a Lie group $G$ with a Lie algebra $\mathfrak{g}=\mathrm{T}_eG$. 
Recall that a (left) {\it Haar measure} on~$G$ is given by a left-invariant volume form on~$G$ \cite{Ja98}. 
Every Lie group admits a Haar measure given by a left-invariant volume form, and it is unique up to a non-zero multiplicative constant 
(cf.\ \cite{Bu13}).

Let $X^L_1,\ldots, X_r^L$ be a basis of~the Lie algebra $\mathfrak{X}_L(G)$ of left-invariant vector fields on~$G$ and
let
$\eta^L_1,\ldots,\eta_r^L$ 
be the dual basis of left-invariant differential 1-forms.
Then any left-invariant volume form on~$G$ is a nonzero scalar multiple of
$$
\Theta = \eta_1^L \wedge\ldots\wedge \eta_r^L
\,.
$$
If $X^L$ is any left-invariant vector field on~$G$, then
\begin{equation}
\label{Eq:exp}
\mathcal{L}_{X^L} \Theta = -{\rm Tr}({\rm ad}_{X^L}) \,\Theta
\,;
\end{equation}
here ${\rm Tr}$ denotes the trace of an endomorphism,
and 
${\rm ad}\colon \mathfrak{g} \to {\rm End}(\mathfrak{g})$,
$v \mapsto {\rm ad}_v$,
denotes the adjoint representation of a Lie algebra~$\mathfrak{g}$,
given by 
${\rm ad}_v w = [v,w]$.

Remember that a Lie group is called {\it unimodular} if its Haar measure is also right-invariant \cite{Mi76}.
All Abelian Lie groups, as well as all compact and semi-simple Lie groups, are unimodular \cite{Yo15}. 
In this work we are mainly concerned with the Lie algebras of unimodular Lie groups, whose main properties are detailed in the following definition and proposition.

\begin{definition}
\label{def:unimodular_lie_algebra}
A finite-dimensional Lie algebra $\mathfrak{g}$ is called {\it unimodular} when the maps 
${\rm ad}_v \in {\rm End}(\mathfrak{g})$
are traceless
---we say that the adjoint representation is {\it traceless}.
\end{definition}

\begin{proposition}
\label{prop:unimodular_LG}
A (connected) Lie group $G$ is unimodular if and only if its Lie algebra is unimodular.
\end{proposition}
\begin{proof} 
Let $\Theta$ be a left-invariant volume form on~$G$.
The Lie algebra is unimodular if and only if
the adjoint representation is traceless,
namely 
the right-hand side of equation (\ref{Eq:exp}) 
is zero for any $X^L \in \mathfrak{X}_L(G)$.
But this amounts to saying that $\Theta$ is also right-invariant:
a tensor field $T$ on a connected Lie group~$G$ is right-invariant 
if and only if
it is invariant with respect to all left-invariant vector fields.
\end{proof}

\begin{remark}
\label{Rem:Important}
A comment about the proof of Proposition \ref{prop:unimodular_LG} is  pertinent.
It is well-known that each
left-invariant vector field on $G$ admits a flow of the form $\phi:t\in\mathbb{R}\mapsto R_{\exp(tv)}\in {\rm Diff}(G)$ for a certain $v\in \mathfrak{g}$.
From this it follows that
a vector field $Y$ on a connected Lie group~$G$ is right-invariant 
if and only if it commutes with \emph{every} left-invariant vector field~$X$, namely
$\mathcal{L}_X Y = 0$.
This also applies to tensor fields on~$G$.
\end{remark}

\section{Multisymplectic Lie systems}\label{section:multi_LS}

This section shows that there exist physical models whose dynamic can be studied through Lie systems admitting a Vessiot--Guldberg Lie algebra of Hamiltonian vector fields relative to a multisymplectic structure. 
This suggests us to introduce the hereafter called {\it multisymplectic Lie systems}. 
Next, the most fundamental properties of these systems are detailed.

\subsection{Example: Schwarz equation}
\label{subsection:KSeq}

Consider a Schwartz equation \cite{Be07,LG99,OT09} of the form
\begin{equation}\label{Eq:KS3}
\frac{d^3x}{dt^3}=\frac 32\left(
\frac{dx}{dt}\right)^{-1}\!\!\left(\frac{d^2x}{dt^2}\right)^{2}\!\!+2b_1(t)\frac{dx}{dt}.
\end{equation}
The relevance of this differential equation is due to its appearance in the study of Milne--Pinney equations and the Schwarz derivative (see \cite{BHLS15,CGL12,CGLS14} and references therein). 

The differential equation (\ref{Eq:KS3}) is known to be a higher-order Lie system \cite{CGL12}.
This means that the associated system of first-order differential equations obtained by adding the variables
$v= dx/dt$ and $a= d^2x/dt^2$, i.e.
\begin{equation}\label{Eq:firstKS3}
\frac{dx}{dt}=v,\qquad \frac{dv}{dt}=a,\qquad \frac{da}{dt}=\frac 32 \frac{a^2}v+2b_1(t)v\,,
\end{equation}	
is a Lie system.
Indeed, it is associated with the $t$-dependent vector field on $\mathcal{O}=\{(x,v,a)\in\mathbb{R}^3\mid v\neq 0\}$ of the form
\begin{equation}\label{Eq:Ex}
{X^{S}}=X_3+b_1(t)X_1,
\end{equation}
where the vector fields given by
\begin{equation}\label{Eq:VFKS1}
\begin{array}{c}
X_1=2v\dfrac{\partial}{\partial a},\qquad X_2=v\dfrac{\partial}{\partial v}+2a\dfrac{\partial}{\partial a},\qquad X_3=v\dfrac{\partial}{\partial x}+a\dfrac{\partial}{\partial v}+\dfrac 32
\dfrac{a^2}v\dfrac{\partial}{\partial a},\end{array}
\end{equation}
satisfy
the commutation relations
\begin{equation}\label{Eq:KSbracket}
[X_1,X_2]=X_1,\quad [X_1,X_3]=2X_2, \quad [X_2,X_3]=X_3.
\end{equation}
As a consequence, $X_1, X_2,$ and $X_3$ 
span a three-dimensional Lie algebra of vector fields $V^{S}$ 
isomorphic to $\mathfrak{sl}_2$, 
and $X^{S}$ becomes a $t$-dependent vector field taking values in $V^{S}$, i.e. 
$(\mathcal{O},X^{S},V^S)$ is a Lie system.

Let us determine a multisymplectic structure on $\mathcal{O}$ so that the vector fields of $V^{S}$ become locally Hamiltonian relative to it. Since $X_1,X_2,X_3$ span ${\rm T}\mathcal{O}$,
then $X_1\wedge X_2\wedge X_3\neq 0$ and $X_1,X_2,X_3$ admit
a family of dual forms $\eta_1,\eta_2$, and $\eta_3$,
i.e. $\eta_\alpha(X_\beta)=\delta_{\alpha}^{\beta}$ for $\delta_\alpha^\beta$ being the {\it Kronecker delta function} and $\alpha,\beta=1,2,3$.
In local coordinates,
$$
\eta_1=\frac{a^2}{4v^3}dx-\frac a{v^ 2}dv+\frac 1{2v} da,\qquad \eta_2=-\frac{a}{v^2}dx+\frac 1vdv,\qquad \eta_3=\frac{1}{v}dx.
$$
With them one can construct the volume form
\begin{equation}
\label{Eq:mult1}
\Theta_S = \eta_1\wedge\eta_2\wedge\eta_3 =
\frac{1}{2v^3} \, da \wedge dv \wedge dx
\,.
\end{equation}
It is well known that,
if a frame $X_\alpha$ of $TN$ satisfies
$[X_\alpha,X_\beta] = c_{\alpha\beta}^\gamma \,X_\gamma$,
then the Lie derivatives of the dual frame $\eta_\alpha$ are given by
$
\mathcal{L}_{X_\alpha} \eta_\gamma =
- c_{\alpha\beta}^{\gamma} \,\eta_\beta \,.
$
In our example we have
$$
\begin{aligned}
&\mathcal{L}_{X_1}\eta_1=-\eta_2 \,, \qquad &\mathcal{L}_{X_1}\eta_2&=-2\eta_3 \,,\qquad 
&\mathcal{L}_{X_1}\eta_3&=0 \,,
\\
&\mathcal{L}_{X_2}\eta_1=\eta_1 \,,\qquad &\mathcal{L}_{X_2}\eta_2&=0 \,,\qquad 
&\mathcal{L}_{X_2}\eta_3&=-\eta_3 \,,
\\
&\mathcal{L}_{X_3}\eta_1=0 \,,\qquad &\mathcal{L}_{X_3}\eta_2&=2\eta_1 \,,\qquad 
&\mathcal{L}_{X_3}\eta_3&=\eta_2 \,,
\end{aligned}
$$
from which it is easily proved that
\begin{equation}
\label{Eq:con}
\mathcal{L}_{X_\alpha}\Theta_S=0 \,,
\qquad \alpha=1,2,3.
\end{equation}
This proves that the $X_\alpha$ are locally Hamiltonian with respect to $\Theta_S$. 
But indeed
\begin{equation}\label{Eq:HamSc}
\begin{gathered}
\iota_{X_1}\Theta_S= \frac 1{v^2}dv\wedge dx =
-d\eta_3,
\\
\iota_{X_2}\Theta_S = \frac 1{v}\left(\frac a{v^2}dv-\frac 1{2v} da\right)\wedge dx =
\frac 12d\eta_2,
\\
\iota_{X_3}\Theta_S =
-\frac{3a^2}{4v^4}dx\wedge dv-\frac{a}{2v^3}da\wedge dx+\frac 1{2v^2}da\wedge dv =
-d\eta_1
\,;
\end{gathered}
\end{equation}
therefore, $X_1,X_2,$ and $X_3$ are Hamiltonian vector fields with respect to the multisymplectic structure
$(\mathcal{O},\Theta_S)$,  with Hamiltonian one-forms
$\theta_1=-\eta_3,\,\, 
\theta_2=\frac 12\eta_2,\,\, 
\theta_3=-\eta_1$.

As a consequence of the above, independently of the $t$-dependent coefficients in (\ref{Eq:firstKS3}), the evolution of $X^{S}$ preserves the volume form $\Theta_S$.
Since $\mathcal{D}^{V^{S}}={\rm T}\mathcal{O}$ and in view of (\ref{Eq:con}), the value of $\Theta_S$ at a point $o\in\mathcal{O}$ determines the value of $\Theta_S$ on the connected component of $o$ in $\mathcal{O}$.
Moreover, $\Theta_S$ is,
up to a multiplicative constant on each connected component of $\mathcal{O}$,
the only volume form satisfying the equations~(\ref{Eq:con}).
Since every one-form and two-form on a three-dimensional manifold are 1--degenerate, the system under study has a unique,
up to a non-zero proportional constant,
multisymplectic form which is invariant under the action of $V^{S}$.

\subsection{Definition and main properties of multisymplectic Lie systems}
\label{subsection:definition_properties_MLS}

The example given in Section \ref{subsection:KSeq}, 
along with the other multisymplectic Lie systems detailed throughout the rest of this work, motivate the following definition.

\begin{definition}
\label{def:MLS}
A (locally) \textit{multisymplectic Lie system} is a triple $(N,\Theta,X)$, where $X$ is a Lie system whose smallest Lie algebra $V^{X}$ is a finite-dimensional real Lie algebra of (locally) Hamiltonian vector fields relative to a multisymplectic structure $\Theta$ on~$N$.
If $\Theta$ has degree~$k$, we say that 
$(N,\Theta,X)$ is a {\it multisymplectic Lie system of degree~$k$}.
\end{definition}

In view of the above, the Schwarz equation (written as a first-order system) defines a multisymplectic Lie system
$(\mathcal{O},\Theta_S,X^{S})$
of degree~3. 

A relevant family of multisymplectic Lie systems is provided by automorphic Lie systems,
as stated by the following proposition:

\begin{proposition}
Every automorphic Lie system $(G,X,V^R)$, 
where $V^R$ is the Lie algebra of right-invariant vector fields on a connected Lie group~$G$, 
is a locally multisymplectic Lie system relative to 
any left-invariant volume form.
\end{proposition}
\begin{proof}
The Lie system $X$ has as a Vessiot--Guldberg Lie algebra 
the set of {right}-invariant vector fields. 
Therefore, any {left}-invariant differential form on~$G$
is invariant with respect to them.
In particular, this is true for a left-invariant volume form~$\Theta$,
which is also a multisymplectic form.
From $\mathcal{L}_{X^R} \Theta = 0$,
the $X^R \in V^R$ are locally Hamiltonian vector fields with respect to~$\Theta$.
\end{proof}

Next it will be shown that every multisymplectic Lie system is related to a finite-dimensional Lie algebra of Hamiltonian forms induced by a Vessiot--Guldberg Lie algebra. 
This Lie algebra will be a key structure for the determination of superposition rules for multisymplectic Lie systems.

\begin{definition} 
Let $(N,\Theta,X)$ be a multisymplectic Lie system
of degree $k$.
A {\it Lie--Hamilton differential form} for it is a $t$-dependent Hamiltonian differential form $\theta$ for~$X$, i.e. 
$\iota_{X_t} \Theta = d \theta_t$ for every~$t$. 
A {\it Lie--Hamilton algebra} for the system is a finite-dimensional Lie algebra of differential forms of degree $k{-}1$ 
(relative to the Lie bracket (\ref{braform})) 
containing all the differentials $d \theta_t$. 
\end{definition}

Recall that every locally Hamiltonian vector field relative to a multisymplectic form of degree $k$ gives rise to a closed differential $(k{-}1)$-form and this correspondence is an injective Lie algebra anti-homorphism relative to the Lie bracket of vector fields and the Lie bracket of  differential $(k{-}1)$-equations of the form (\ref{braform}). 
One has the following trivial consequence:
\begin{proposition}
\label{prop:LHalgebra}
Every multisymplectic Lie system $(N,\Theta,X)$ possesses a \emph{minimal} Lie--Hamilton algebra, namely
$$
\mathfrak{M}=\{ \iota_Z \Theta \mid Z \in V^X \}
\,.
$$
\end{proposition}

\paragraph{Example}
As seen in Section \ref{subsection:KSeq},
the multisymplectic Lie system $(\mathcal{O},\Theta_S,X^S)$ 
associated with the Schwarz equation
is such that the vector fields $X_1,X_2,X_3$
spanning its Vessiot--Guldberg Lie algebra, $V^S$,
are Hamiltonian with Hamiltonian forms
\begin{equation}\label{FormSE}
\theta_1 = -\frac{1}{v}dx,
\quad
\theta_2 = -\frac{a}{2v^2}dx+\frac{1}{2v}dv,
\quad 
\theta_3 = -\frac{a^2}{4v^3}dx+\frac{a}{v^2}dv-\frac{1}{2v}da \,.
\end{equation}
Therefore, their differentials
\begin{equation}\label{dfSE}
d\theta_1= \frac{dv\wedge dx}{v^2},
\quad 
d\theta_2= \frac{adv\wedge dx}{v^3} - \frac{da\wedge dx}{2v^2},
\quad 
d\theta_3= \frac{-3a^2dx\wedge dv}{4v^4} - \frac{ada\wedge dx}{2v^3} + \frac{da\wedge dv}{2v^2}
\end{equation}
satisfy the commutation relations 
(for the bracket of differentials of Hamiltonian forms)
\begin{equation}\label{Eq:Sch-M}
\{d\theta_1,d\theta_2\} = -d\theta_1,
\quad 
\{d\theta_1,d\theta_3\} = -2d\theta_2,
\quad 
\{d\theta_2,d\theta_3\} = -d\theta_3 \,,
\end{equation}
and span a Lie--Hamilton algebra $\mathfrak{M}$ of the system.\demo

\section{Locally automorphic Lie systems and invariant forms}
\label{section:LS-IF}

In this section we analyse conditions 
under which one can ensure that 
there is a multisymplectic form~$\Theta$ 
invariant with respect to the elements of a Vessiot--Guldberg Lie algebra~$V$.

It may be difficult to find multisymplectic forms compatible with a Lie system $X$ admitting a Vessiot--Guldberg Lie algebra~$V$ 
as this requires to search for appropriate solutions, namely $\Theta$, 
of a system of partial differential equations 
$\mathcal{L}_Y\Theta=0$ for every $Y\in V$. 
Nevertheless, we can devise several simpler methods to find compatible invariant forms 
for a particular class of Lie systems with relevant physical applications: 
the hereafter locally automorphic Lie systems.

\subsection{Locally automorphic Lie systems}

\begin{definition} 
A {\it locally automorphic Lie system} on $N$ is a triple $(N,X,V)$, 
where $X$ is a Lie system on $N$ with a Vessiot--Guldberg Lie algebra~$V$ 
such that $\dim V=\dim N$ and $\mathcal{D}^V={\rm T}N$.
\end{definition}

Locally automorphic Lie systems are called in this way because they are locally diffeomorphic to automorphic Lie systems. 
The following theorem proves this fact.

\begin{theorem}
\label{Trivial}
Let $(N,X,V)$ be a locally automorphic Lie system,
let $G$ be a Lie group whose Lie algebra is isomorphic to~$V$,
let $\varphi$ be a local action of~$G$ on~$N$
obtained from the integration of~$V$,
and let $(G,X^G,V^G)$ be the corresponding automorphic Lie system on~$G$ given by Theorem \ref{Th:XG-X}.
For every $x \in N$ the map 
$\varphi_x = \varphi(\cdot,x)$
is a local diffeomorphism mapping $X^G$ to~$X$.
\end{theorem}
\begin{proof}
Recall that $N$ is assumed to be an $n$-dimensional manifold.
As stated in Theorem \ref{Th:XG-X},
given $x \in N$,
the map $\varphi_x$ relates a basis  of right-invariant vector fields 
$X_\alpha^R$ of~$G$ with a basis
$X_\alpha$ of~$V$, i.e. $T_g \varphi_{x} \colon T_gG \to T_{\varphi(g,x)}N$
maps
$X_\alpha^R(g)$ onto $X_\alpha(\varphi(g,x))$.
But due to the definition of locally automorphic system,
the $X_\alpha^R(g)$ are $n$ linearly independent vectors 
constituting a basis of $T_gG$,
and the $X_\alpha(\varphi(g,x))$ are also $n$ linearly independent vectors
constituting a basis of $T_{\varphi(g,x)}N$.
Thus, $T_g \varphi_{x}$ is a linear isomorphism,
and $\varphi_{x}$ becomes a local diffeomorphism.
Therefore,
when $\varphi_x$ is restricted to open sets yielding a diffeomorphism,
it sends $X^G$ onto~$X$.
\end{proof}

Remember that the action $\varphi$ can be ensured to be globally defined  only if $G$ is simply connected and $V$ consists of complete vector fields 
\cite{Palais}. 
For simplicity, we will hereafter assume that $\varphi$ is globally defined. 

The mapping $\varphi$ not only allows us to establish local diffeomorphisms $\varphi_x:G\rightarrow N$, with $x\in N$, 
but also maps certain geometric structures related to the locally automorphic Lie system $(N,X,V)$ with the associated automorphic one $(G,X^R,V^R)$. 


In view of Theorem \ref{Trivial}, one obtains the following corollaries.


\begin{corollary}
\label{cor:superposition_rule_locally_automorphic_LS}
Let $(N,X,V)$ be a locally automorphic Lie system. 
Then, $X$ admits a superposition rule depending on only one particular solution of~$X$.
\end{corollary}
\begin{proof}
This is a consequence of Theorem \ref{Trivial} and the fact that automorphic Lie systems admit a superposition rule depending on one particular solution, 
as mentioned at the end of Section~\ref{subsection:Lie systems}.
\end{proof}

\begin{corollary}
\label{cor:superposition_rule_locally_automorphic_LS_2}
If $(N,X,V^X)$ is a locally automorphic Lie system on a (connected) manifold, 
then all $t$-independent constants of motion of $X$ are constants.
\end{corollary}
\begin{proof}
If $f$ is a constant of motion, then
for every $Z \in V^X$ we have $\mathcal{L}_Z f = 0$.
As by hypothesis the vector fields of $V^X$ span ${\rm T}_xN$ at every $x \in N$,
one has that $df=0$ and $f$ is locally constant. 
If the manifold is connected, then $f$ constant.
\end{proof}

The existence of the local diffeomorphism $\varphi_x$ 
is useful to obtain theoretical properties 
of locally automorphic Lie systems,
as for instance Corollary \ref{cor:superposition_rule_locally_automorphic_LS}.
Nevertheless,
its practical use to map them into automorphic Lie systems is quite limited
because the locality of $\varphi_x$ and the difficulties to obtain an explicit expression, which must be obtained by solving a system of nonlinear ordinary differential equations determined by $V$.
This is illustrated in the following two examples
of locally automorphic Lie systems. 

\begin{example} 
\textbf{(The generalized Darboux--Brioschi--Halphen (DBH) system)} Consider the system of differential equations  
\cite{Darboux}:
\begin{equation}
\label{Eq:Partial}
\begin{gathered}
\left\{
\begin{aligned}
\frac{{\rm d}w_1}{{\rm d}t}& = w_3w_2-w_1w_3-w_1w_2+\tau^2,\\
\frac{{\rm d}w_2}{{\rm d}t}& = w_1w_3-w_2w_1-w_2w_3+\tau^2,\\
\frac{{\rm d}w_3}{{\rm d}t}& = w_2w_1-w_3w_2-w_3w_1+\tau^2,\\
\end{aligned}\right.\\
\tau^2=
\alpha_1^2(\omega_1-\omega_2)(\omega_3-\omega_1)+
\alpha_2^2(\omega_2-\omega_3)(\omega_1-\omega_2)+
\alpha_3^2(\omega_3-\omega_1)(\omega_2-\omega_3),
\\
\alpha_1,\alpha_2,\alpha_3\in \mathbb{R}.
\end{gathered}
\end{equation}
The DBH system with $\tau=0$ appears in the description of triply orthogonal surfaces and the vacuum Einstein equations for hyper-K\"ahler Bianchi-IX metrics \cite{CH03,Darboux,Halphen}. 
Meanwhile, the generalized DBH system for $\tau\neq 0$ is a
reduction of the self-dual Yang--Mills equations corresponding to an infinite-dimensional gauge group of diffeomorphisms of a
three-dimensional sphere \cite{CH03}. 

Although the DBH system is autonomous, 
it is more appropriate, 
e.g.\ to obtain its Lie symmetries \cite{EHLS-2014}, 
to consider it as a Lie system related to a Vessiot--Guldberg Lie algebra $V^{\rm DBH}$ spanned by
$$
\begin{gathered}
X_1^{\rm DBH} =
\frac{\partial}{\partial w_1} + \frac{\partial}{\partial w_2} + \frac{\partial}{\partial w_3},
\quad 
X_2^{\rm DBH} =
w_1\frac{\partial}{\partial w_1} + w_2\frac{\partial}{\partial w_2} + w_3\frac{\partial}{\partial w_3}
\\
X_3^{\rm DBH} =
-(w_3w_2-w_1(w_3+w_2)+\tau^2)\frac{\partial}{\partial w_1} - (w_1w_3-w_2(w_1+w_3)+\tau^2)\frac{\partial}{\partial w_2}
\\ 
- (w_2w_1-w_3(w_2+w_1)+\tau^2)\frac{\partial}{\partial w_3}.
\end{gathered}
$$
In fact, 
$$
[X_1^{\rm DBH},X_2^{\rm DBH}] = X^{\rm DBH}_1,
\qquad 
[X_1^{\rm DBH},X_3^{\rm DBH}] = 2X_2^{\rm DBH},
\qquad 
[X_2^{\rm DBH},X_3^{\rm DBH}] = X_3^{\rm DBH}.
$$
Hence, 
$\dim V^{\rm DBH}=\dim \mathcal{O}$ and 
$X_1^{\rm DBH}\wedge X_2^{\rm DBH}\wedge X_3^{\rm DBH}\neq 0$ 
on an open submanifold $\mathcal{O}$ of $\mathbb{R}^3$. 
Thus, $(\mathcal{O},X^{DBH}_3,V^{\rm DBH})$ is a locally automorphic Lie system. 
To obtain a local diffeomorphism mapping this system into an automorphic one, 
we need to integrate the vector fields of $V^{\rm DBH}$. 
Their analytic form makes it clear that it is very hard to provide such a local diffeomorphism. \demo
\end{example}

\begin{example}
\label{exampleCS}
{\bf (\textbf{A control system}
\cite{Ni00,Ra06})}. 
Consider the system of differential equations on $\mathbb{R}^5$ given by
\begin{equation}
\label{Eq:ControlSys}
\begin{gathered}
\frac{{\rm d}x_1}{{\rm d}t}=b_1(t),\quad 
\frac{{\rm d}x_2}{{\rm d}t}=b_2(t),\quad
\frac{{\rm d}x_3}{{\rm d}t}=b_2(t)x_1,\quad 
\frac{{\rm d}x_4}{{\rm d}t}=b_2(t)x_1^2, \quad 
\frac{{\rm d}x_5}{{\rm d}t}=2b_2(t)x_1x_2,\
\end{gathered}
\end{equation}
where $b_1(t)$ and $b_2(t)$ are arbitrary $t$-dependent functions. 

This system is defined by the $t$-dependent vector field 
$X^{\rm CS}=b_1(t)X_1+b_2(t)X_2$ on $\mathbb{R}^5$,
where the vector fields
\begin{equation}
\label{Eq:BasisControl}
\begin{gathered}
X_1 = \frac{\partial}{\partial x_1},
\qquad 
X_2 = \frac{\partial}{\partial x_2}+x_1\frac{\partial}{\partial x_3}+x_1^2\frac{\partial}{\partial x_4}+2x_1x_2\frac{\partial}{\partial x_5},
\\[1ex]
X_3 = \frac{\partial}{\partial x_3}+2x_1\frac{\partial}{\partial x_4}+2x_2\frac{\partial}{\partial x_5},
\qquad
X_4 = \frac{\partial}{\partial x_4},
\qquad 
X_5 = \frac{\partial}{\partial x_5},
\end{gathered}
\end{equation}
are such that their only non-vanishing commutation relations read
\begin{equation}
\label{Eq:ConRel}
[X_1,X_2]=X_3,\qquad [X_1,X_3]=2X_4,\qquad [X_2,X_3]=2X_5.
\end{equation}
Hence, $X_1,\ldots,X_5$ span a five-dimensional nilpotent Lie algebra $V^{\rm CS}$. 
Since $X^{\rm CS}$ takes values in $V^{\rm CS}$, 
then $(\mathbb{R}^5,X^{\rm CS},V^{\rm CS})$ is a Lie system as already noticed in \cite{Ra06}. 
The vector fields of $V^{\rm CS}$ span a distribution 
$\mathcal{D}^{V^{\rm CS}}={\rm T}\mathbb{R}^5$ 
and $\dim V^{\rm CS}=\dim \mathbb{R}^5$. 
Hence, 
$(\mathbb{R}^5,X^{\rm CS},V^{\rm CS})$ 
is a locally automorphic Lie system. \demo
\end{example}

\subsection{Invariants for locally automorphic Lie systems}\label{Sec:Inv}

Recall that a Lie symmetry of a Lie system $(N,X,V)$
is a vector field $Y$ on~$N$ such that
$\mathcal{L}_{Y} Z =0$
for every vector field $Z \in V$.
If $V = \langle X_1,\ldots,X_r \rangle$, 
this is equivalent to saying that $Y$ has to satisfy the system of partial differential equations 
\begin{equation}
\label{Eq:SysRos}
\mathcal{L}_{X_i}Y = 0,\qquad i=1,\ldots,r.
\end{equation}
The set ${\rm Sym}(V)$
of Lie symmetries of $(N,X,V)$ is a Lie algebra. 
Let us study this set for the case of locally automorphic Lie systems.

If $(N,X,V)$ is a locally automorphic Lie system, then each mapping $\varphi_x$  maps it onto an automorphic Lie system $(G,X^R,V^R)$. It is immediate that ${\rm Sym}(V^R)=V^L$. Since $\varphi_x$ is a local diffeomorphism mapping $V$ onto $V^R$, then it also maps ${\rm Sym}(V)$ onto $V^L$. Hence, one obtains the following lemma, whose implications will be illustrated in Example \ref{Sym}.

\begin{lemma}
\label{lemma:Exis}
Let $(N,X,V)$ be a locally automorphic Lie system. 
The Lie algebra ${\rm Sym}(V)$ of symmetries of~$V$ is isomorphic to~$V$. 
\end{lemma}

\begin{example}\label{Sym}
We reconsider Example \ref{exampleCS}
studying the control system given by (\ref{Eq:ControlSys}).
By solving the linear system of partial differential equations (\ref{Eq:SysRos}) in the unknown coefficients of $Y$ in the basis $\partial/\partial_{x_1},\ldots,\partial/\partial_{x_5}$, which demands a very long and tedious calculation, 
one gets that every Lie symmetry $Y$ of an arbitrary control system (\ref{Eq:ControlSys}) must be a linear combination with constant coefficients of the vector fields
\begin{equation}
\label{Eq:SymCS}
\begin{gathered}
Y_1 = \frac{\partial}{\partial x_1} + x_2\frac{\partial }{\partial x_3} + 2x_3\frac{\partial}{\partial x_4} + x_2^2\frac{\partial}{\partial x_5}, \qquad
Y_2 = \frac{\partial}{\partial x_2} + 2x_3\frac{\partial }{\partial x_5}, \\[1ex]
Y_3 = \frac{\partial}{\partial x_3}, 
\qquad
Y_4 = \frac{\partial}{\partial x_4}, 
\qquad
Y_5 = \frac{\partial}{\partial x_5}.
\end{gathered}
\end{equation}
A straightforward calculation shows that the vector fields $-Y_i$, with $i=1,\ldots,5$, 
generate a Lie algebra with the same structure constants as $X_1,\ldots,X_5$. \demo

\end{example}

Since every locally automorphic Lie system $(N,X,V)$ is locally diffeomorphic to an automorphic Lie system $(G,X^R,V^R)$, 
one has that every differential form on $N$ invariant with respect to the Lie derivative of elements of~$V$ 
must be locally diffeomorphic to a left-invariant differential form on~$G$. 
Since ${\rm Sym}(V)$ is also diffeomorphic to $V^L$, 
in view of Remark \ref{Rem:Important}
one obtains the following theorem: 
\begin{theorem}
\label{th:invariant_form} 
Let $(N,X,V)$ be a locally automorphic Lie system 
and let $Y_1,\ldots, Y_r$ be a basis of ${\rm Sym}(V)$, 
with dual frame $\nu^1,\ldots, \nu^r$. 
Then, a differential form on~$N$ is invariant with respect to the 
Lie algebra $V$
if and only if
it is a linear combination with real coefficients of exterior products of $\nu^1,\ldots, \nu^r$.
\end{theorem}

\begin{example}\nonumber
We consider again Example \ref{exampleCS}
studying the control system given by (\ref{Eq:ControlSys}).
Its Lie symmetries are given by the vector fields
$Y_1,\ldots,Y_5$ described in (\ref{Eq:SymCS}),
whose dual frame reads
\begin{equation}
\label{example:control_theory_dual_forms}
\begin{gathered}
\eta_1 = dx_1,\quad
\eta_2 = dx_2,\quad
\eta_3 = -x_2dx_1+dx_3 \,,
\\
\eta_4 = -2x_3dx_1+dx_4,\quad
\eta_5 = -{x^2_2}dx_1-2x_3dx_2+dx_5 \,.
\end{gathered}
\end{equation}
Therefore, according to the preceding theorem,
all the invariant differential forms of this control system
are linear combinations, with real coefficients,
of the exterior products of the $\eta_i$;
this space is
$\mathsf{\Lambda} ({\rm Sym}(V^{\rm CS}))^*$.%

It is easy to obtain multisymplectic forms invariant with respect to 
$V^{\rm CS}$ within this set. 
It follows from linear algebra considerations that all non-degenerate differential forms on $\mathbb{R}^5$ must have rank five or three. 
For instance, if we consider the invariant $5$-form
\[
\Theta_{vol} = \eta_1 \wedge \eta_2 \wedge \eta_3 \wedge \eta_4 \wedge \eta_5,
\]
we obtain a left-invariant volume form, Hence, this is a multisymplectic form satisfying that
\[
\mathcal{L}_Y\Theta_{vol}=0\,,\quad 
\forall Z\in V^{\rm CS}\,.
\]
Therefore, $(\mathbb{R}^5,\Theta_{vol}, X^{\rm CS})$ is a multisymplectic Lie system.

We can obtain other examples of multisymplectic forms compatible with this Lie system. 
Since
$$
d\eta_1=0,\quad d\eta_2=0,\quad d\eta_3=\eta_1\wedge \eta_2,\quad d\eta_4=2\eta_1\wedge \eta_3,\quad d\eta_5=2\eta_2\wedge \eta_3,
$$
we consider now the closed $3$-form
\begin{multline}
\label{Eq:example_control_theory_multi_form}
\Theta=
d(\eta_3\wedge \eta_4)+d(\eta_4\wedge \eta_5)
=\eta_1\wedge \eta_2\wedge \eta_4+2\eta_1\wedge\eta_3\wedge\eta_5-2\eta_4\wedge \eta_2\wedge \eta_3
\\
=(1-2x_2)dx_{124}+8x_3dx_{123}+2dx_{135}-2dx_{234}
\,,
\end{multline}
where we use the notation
$d x_{ijk} = d x_i \wedge d x_j \wedge d x_k$.
It is easy to prove that this $3$-form is non-degenerate, then it is a multisymplectic form of degree~$3$. 
Therefore $(\mathbb{R}^5,\Theta,X^{\rm CS})$ is a new  multisymplectic Lie system. 
It is worth noting that $\Theta$ is not a volume form. \demo
\end{example}

Although Lemma \ref{lemma:Exis} guarantees the existence of the Lie algebra ${\rm Sym}(V)$, 
its computation may become computationally complicated. 
For this reason, we will provide Theorem \ref{th:MTS}, 
which gives a family of invariant differential forms 
under the action of the elements of a Lie algebra~$V$ 
of a locally automorphic Lie system 
without using the Lie algebra of Lie symmetries 
${\rm Sym}(V)$.

We need to introduce previously some additional structures.
First,
every linear representation 
$\rho\colon \mathfrak{g} \rightarrow {\rm End}(E)$ 
of a Lie algebra $\mathfrak{g}$
on a vector space~$E$ 
can be extended to a linear representation
on the exterior algebra of~$E$;
its elements are indeed derivations, so this yields a homomorphism
$\mathfrak{g} \rightarrow {\rm Der}(\mathsf{\Lambda} E)$ 
\cite[p.\,110]{Gr78}. 
Let us consider the adjoint representation
${\rm ad}_v$ of~$\mathfrak{g}$, 
and denote by
${\rm coad}_v$
its contragradient representation,
which is the linear representation on the dual space
$\mathfrak{g}^*$
given by
${\rm coad}_v = -({\rm ad})^\top$.
We apply the preceding remark to this representation,
thus obtaining a map
$\mathfrak{g} \rightarrow 
{\rm Der}(\mathsf{\Lambda} \mathfrak{g}^*)$
that we denote by 
$v \mapsto D_v$.

\begin{theorem}
\label{th:MTS} 
Let  $(N,X,V)$ be a locally automorphic Lie system,
and let 
$\phi\colon \mathfrak{g} \rightarrow V$ be 
a Lie algebra isomorphism.
Then we have:
\begin{enumerate}
\item
The isomorphism 
$\phi\colon \mathfrak{g} \rightarrow V$
maps the adjoint endomorphism ${\rm ad}_v$ of~$\mathfrak{g}$
to the Lie derivative $\mathcal{L}_{\phi(v)}$
of the vector fields in~$V$.
\item
The dual space $V^*$
can be identified with the set 
$\{ \theta \in \Omega^1(N) \mid 
\forall X \in V,\, \iota_X\theta
\hbox{ is constant} \}$.
\\
With this identification,
the contragradient isomorphism 
$\phi^\vee \colon \mathfrak{g}^* \rightarrow V^*$
also maps the coadjoint endomorphism ${\rm coad}_v$
to the Lie derivative $\mathcal{L}_{\phi(v)}$ of 1-forms.
\item
The exterior power $\mathsf{\Lambda}^p V^*$
can be identified with the set of differential $p$-forms 
$\theta$ on~$N$
whose contractions with $p$ vector fields of~$V$ are constant.
\\
With this identification,
the prolongation of the contragradient isomorphism
to the exterior algebras,
$\mathsf{\Lambda} \phi^\vee \colon 
\mathsf{\Lambda} \mathfrak{g}^* \rightarrow \mathsf{\Lambda} V^*$,
maps the derivation $D_v$ to 
the Lie derivative $\mathcal{L}_{\phi(v)}$ of differential forms.
\item
If an element $\omega \in \mathsf{\Lambda} \mathfrak{g}^*$
satisfies that 
$D_v \omega=0$ for each $v \in \mathfrak{g}$,
then its image, 
$\mathsf{\Lambda} \phi^\vee (\omega)$, is a differential form in~$N$
which is invariant with respect to the elements of~$V$;
namely, the Lie derivative of $\mathsf{\Lambda}\phi^\vee(\omega)$ 
relative to elements of $V$ vanishes.
\end{enumerate}
\end{theorem}
\begin{proof}
The first assertion is immediate:
$
\phi({\rm ad}_v u) = 
\phi([v,u]) =
[\phi(v),\phi(u)] =
\mathcal{L}_{\phi(v)} \phi(u)
$.

For the second,
if $X_i$ is a basis of~$V$,
it is also a frame of the tangent bundle $TN$;
then, if $\theta^i$ is its dual frame,
the $\theta^i$ can be considered as a basis of~$V^*$.
On the other hand, every 1-form on~$N$ can be written as
$\theta = \sum g_i \theta^i$,
and the functions $g_i$ are constant if and only if,
for every $X \in V$, one has that
$\iota_X\theta$ is constant.
The Lie derivative of 1-forms satisfies 
$
\mathcal{L}_{Y} \iota_X\theta =
\iota_X\mathcal{L}_{Y}\theta +
\iota_{\mathcal{L}_{Y} X}\theta
$.
When $X,Y \in V$ and $\theta \in V^*$ this implies that
$
\iota_X\mathcal{L}_{Y} \theta =
-\iota_{\mathcal{L}_{Y} X} \theta
$,
which means that 
$\mathcal{L}_{Y}$ on~$V^*$
is the minus the transpose of 
$\mathcal{L}_{Y}$ on~$V$;
this completes the proof of the second statement.

Third statement proceeds in a similar way. 
The correspondence between 
$D_v$ and $\mathcal{L}_{\phi(v)}$ 
is a consequence of the fact that 
in their respective algebras both operators are derivations,
and, by the preceding statement, 
they agree when applied to the subspaces 
$\mathfrak{g}^*$ and $V^*$, 
both of which generate the corresponding exterior algebras.

From this we obtain
$
\mathsf{\Lambda} \phi^\vee (D_v \omega) =
\mathcal{L}_{\phi(v)} \,\mathsf{\Lambda} \phi^\vee (\omega)
$.
Therefore, 
the invariance of $\omega$ with respect to every $v \in \mathfrak{g}$
implies the invariance of 
$\mathsf{\Lambda} \phi^\vee (\omega)$
with respect to every $Y \in V$.
\end{proof}


Let us show now that certain conditions on~$X$ allow us to easily construct a multisymplectic form turning the Vessiot--Gulbderg Lie algebra for~$X$ into locally $\Omega$-Hamiltonian vector fields.

The idea is to find $\mathfrak{g}$-invariant elements in $\mathsf{\Lambda}\mathfrak{g}^*$ 
whose image under $\mathsf{\Lambda}\phi^\vee$ is a multisymplectic form. 
In particular, every unimodular Lie algebra gives rise to an invariant element of $\mathsf{\Lambda}\mathfrak{g}^*$ of maximal degree 
whose representation is the volume differential form.

Recall that the local diffeomorphism $\varphi_x$ 
that maps $(N,X,V)$ onto an automorphic system $(G,X^R,V^R)$
maps also $V^*$ onto the $(V^R)^*$, 
which consists of the right-invariant differential one-forms on $G$. 
In view of this, one immediately obtains the following corollary.

\begin{corollary}
\label{cor:InvVolumeForm} 
Let $(N,X,V)$ be a locally automorphic Lie system. 
If $V$ is unimodular, then
$V$ admits an invariant volume form given by
$$
\Theta=\eta^1\wedge\ldots\wedge\eta^r,
$$
where $\eta^1,\ldots,\eta^r$ is any basis of elements of~$V^*$. 
Then, $(N,\Theta,X)$, where $X$ takes values in $V$, 
is a multisymplectic Lie system. 
Moreover, $\Theta$ is invariant with respect to the Lie derivatives with elements of the Lie algebra ${\rm Sym}(V)$.
\end{corollary}

\begin{remark}
Theorem \ref{th:MTS} and Corollary \ref{cor:InvVolumeForm} give us a procedure to easily construct a compatible multisymplectic form certain class of Lie systems. 
In fact, Example \ref{subsection:KSeq}  
was carried out by using this procedure to obtain a compatible multisymplectic structure.
\end{remark}

\paragraph{Example}
The Schwarz equation, 
whose first-order system $X^S$ is given by (\ref{Eq:firstKS3}), 
admits a Lie algebra of symmetries given by 
(see \cite{LG99,OT05})
\begin{equation}\label{Eq:SymSc2}
Y_1= \frac{\partial}{\partial x}
\,,\qquad
Y_2= x\frac{\partial}{\partial x} + v\frac{\partial}{\partial v} + a\frac{\partial}{\partial a}
\,,\qquad
Y_3= x^2\frac{\partial}{\partial x} + 2vx\frac{\partial}{\partial v} + 2(ax+v^2)\frac{\partial}{\partial a} \,.
\end{equation} 
These are therefore the vector fields commuting with the Lie algebra $V^S=\langle X_1,X_2,X_3\rangle$, where $X_1,X_2,X_3$ are given by (\ref{Eq:firstKS3}).
The dual forms to (\ref{Eq:SymSc2}) read
$$
\eta_1=dx-\frac{2v^2x+ax^2}{2v^3}dv+\frac{x^2}{2v^2}da,\qquad \eta_2=\frac{v^2+ax}{v^3}dv-\frac x{v^2}da,\quad \eta_3=-\frac{a}{2v^3}dv+\frac{1}{2v^2}da.
$$
According to Theorem \ref{th:invariant_form}, 
every differential form invariant with respect to the vector fields $X_1,X_2,X_3$ 
spanning the Vessiot--Guldberg Lie algebra $V^S$ of~$X^S$ 
is a linear combination of the exterior products of  
$\eta_1,\eta_2,\eta_3$. 
For instance,
$$
-2\eta_1\wedge \eta_3 =
\frac{adx\wedge dv+vda\wedge dx+xdv\wedge da}{v^3} \,,
\quad 
2\eta_2\wedge \eta_3 =
\frac{dv\wedge da}{v^3} \,,
$$
are invariant with respect to the elements of~$V^S$. 
In fact, these are indeed the invariant presymplectic forms obtained in \cite{CGLS14,LV15} in an ad-hoc manner.
\demo

\bigskip
It may be difficult to find multisymplectic forms compatible with Lie systems, 
but it turns out to be quite easy to find compatible closed invariant forms.
The knowledge of the latter can be used to easily find compatible multisymplectic forms,
as will be showed in following sections.

\begin{proposition}
Every locally automorphic Lie system $(N,X,V)$
has non-zero closed invariant forms.
\end{proposition}
\begin{proof}
The conditions of the Vessiot--Guldberg Lie algebra allow us to assume that 
$X$ along with $N$ are locally diffeomorphic
to a $t$-dependent right-invariant vector field $X^R$ 
on a connected and simply connected Lie group~$G$.
Hence, every left-invariant differential form is invariant under right-invariant vector fields.
It is obvious that there exist closed left-invariant forms on~$G$, 
e.g.\ volume forms or the differentials of left-invariant forms.
The diffeomorphism between $G$ and $N$ maps these left-invariant forms 
onto invariant closed differential forms on $N$ compatible with~$X$.
\end{proof}

In view of the above, there are plenty of multisymplectic compatible forms.
The main point is to determine closed multilinear forms 
on the Chevalley--Eilenberg cohomology of the Lie algebra,
which is a purely algebraic problem. 
We will not give any precise procedure to construct non-degenerate closed elements of the Chevalley--Eilenberg cohomology. 
Nevertheless, every non-zero decomposable $k$-form 
$\eta\in \mathsf{\Lambda} \mathfrak{g}^*$ 
is such that the rank of the mapping 
$\hat \eta: v\in \mathfrak{g} \mapsto \iota_v\eta$ 
is~$k$ and its kernel has dimension $n-k$. 
This fact suggests that appropriately chosen closed $k$-forms 
will have eventually a zero-dimensional kernel and they will become 1-nondegenerate. 
This will be enough to accomplish our aims in this work.   
Indeed this was used in the control system example
in order to obtain the invariant multisymplectic form of degree~3
given by Eq.\ (\ref{Eq:example_control_theory_multi_form}).

\section{Superposition rules for multisymplectic Lie systems}
\label{section:superposition_rules}
	
Let us employ the multisymplectic form of a multisymplectic Lie system $(N,\Theta,X)$ so as to construct superposition rules for $X$.  
In short our idea consists in constructing an abstract tensor algebra through the Lie--Hamilton algebra of Hamiltonian differential forms of $(N,\Theta,X)$. 
Then, we use the algebraic properties of this tensor algebra and their representation as geometric objects 
to obtain invariant tensor fields of the diagonal prolongation
$(N^{[m]},\Theta^{[m]},X^{[m]})$.
From these invariant tensor fields and, eventually, the Lie symmetries of the Lie system $X$,
we will obtain constants of motion 
that will finally lead 
to the superposition rule for~$X$. 
As a byproduct, many other invariants of $X$ appear, e.g.\ symplectic forms invariant under the action of the elements of~$V^X$.

\subsection{Coalgebras, $\mathfrak{g}$-modules, and tensor fields}
\label{subsection:tensor_fields}
	
The following methods rely on considering the structures related to multisymplectic Lie systems as realizations of tensor algebras and $\mathfrak{g}$-modules \cite{Va84}. 
The properties of such algebraic structures
will be then employed to obtain superposition rules for multisymplectic Lie systems. 
We refer to \cite{Bo74,Pressley,Va84} for further details on the algebraic structures appearing in this section.

A {\it coalgebra} is a linear space $A$ along with two mappings $\Delta:A\rightarrow A\boxtimes A$, the {\it coproduct}, and $\epsilon:A\rightarrow \mathbb{R}$, the {\it counit}, such that
\begin{equation}\label{con}
(\Delta\boxtimes {\rm Id}_A)\circ \Delta=
({\rm Id}_A\boxtimes \Delta )\circ \Delta,
\qquad 
({\rm Id}_A\boxtimes \epsilon)\circ \Delta=
(\epsilon\boxtimes {\rm Id}_A)\circ \Delta={\rm Id}_A.
\end{equation}
Here $\boxtimes$ refers to the tensor product used to define a coalgebra and ${\rm Id}_A$ is the identity map. This notation is important and it must not be confused with the usual tensor product $\otimes$, which appears along this section.
	
If $A$ is an associative algebra with unit element, then 
$A^{(m)}=A\boxtimes\stackrel{^m}{\ldots}\boxtimes A$ 
admits a canonical associative algebra structure with unit. 
In particular, 
the tensor algebra $T(\mathfrak{g})$ related to the Lie algebra $\mathfrak{g}$ 
has a unital associative algebra structure that induces a canonical new one in $T^{(m)}(\mathfrak{g})$ 
\cite{Bo74}. 
In view of this, a simple calculation leads to prove the following proposition 
(cf.\ \cite[Ch. 3, Sec. 11]{Bo74}).
\begin{proposition}
The tensor algebra $T(\mathfrak{g})$ admits a coalgebra structure relative to the coproduct, $\Delta$, and the counit, $\epsilon$, given by the unique morphisms of associative algebras satisfying 
$$
\Delta(v) = v\boxtimes 1+1\boxtimes v, \qquad 
\epsilon(v)=0, \qquad 
\forall v\in \mathfrak{g}.
$$
\end{proposition}
	
Analogously, we also define a higher-order coproduct 
$\Delta^{(m)}: 
T(\mathfrak{g}) \rightarrow T^{(m+1)}(\mathfrak{g})$ 
recurrently 
$$
\Delta^{(m+1)}= 
(
{\rm Id}
\boxtimes 
\stackrel{m}{\cdots}
\boxtimes {\rm Id}
\boxtimes \Delta
) \circ 
\Delta^{(m)}, \qquad 
\forall m\in \mathbb{N}, \qquad 
\Delta^{(1)}=\Delta,
$$
which is a morphism of associative algebras.

Recall that a {\it $\mathfrak{g}$-module} is a linear space $E$ along with a Lie algebra morphism 
$\rho:\mathfrak{g}\rightarrow {\rm End}(\mathfrak{g})$.
The adjoint representation 
${\rm ad}: v\in \mathfrak{g} \mapsto {\rm ad}_v\in {\rm Der}(\mathfrak{g})$ 
is such that each ${\rm ad}_v$ can be uniquely extended to a derivation relative to the tensor product in $T(\mathfrak{g})$. 
In this manner, 
$T(\mathfrak{g})$ becomes a $\mathfrak{g}$-module relative to 
${\rm ad}: v \in \mathfrak{g} \mapsto 
{\rm ad}_v\in {\rm Der}(T(\mathfrak{g}))$.
From now on, 
we will denote the adjoint representation of~$\mathfrak{g}$ 
and several of its generalizations and/or extensions 
in the same way as this will not lead to confusion.

The Lie algebra morphism 
${\rm ad}: \mathfrak{g} \rightarrow {\rm Der}(T(\mathfrak{g}))$ 
induces a second one  
${\rm ad}: \mathfrak{g} \rightarrow 
{\rm Der}(T^{(m)}(\mathfrak{g}))$
by requiring
$$
{\rm ad}_v({\bf t}_1\boxtimes\stackrel{m}{\cdots} \boxtimes {\bf t}_m)=
{\rm ad}_v({\bf t}_1) \boxtimes \cdots \boxtimes {\bf t}_m+
\ldots+
{\bf t}_1\boxtimes\cdots\boxtimes{\rm ad}_v({\bf t}_m),
\qquad 
\forall {\bf t}_1,\ldots {\bf t}_m\in T(\mathfrak{g}).
$$
This turns the spaces 
$T^{(m)}(\mathfrak{g})$, with $m\in \mathbb{N}$, 
into $\mathfrak{g}$-modules. 
	
Previous structures have a special relevance for the theory of multisymplectic Lie systems. 
In particular, the space to be defined next plays a significative role 
in the determination of their superposition rules and tensorial invariants.
	
\begin{definition} 
Let $E$ be the $\mathfrak{g}$-module relative to a Lie algebra representation $\rho:\mathfrak{g}\rightarrow {\rm End}(E)$. 
We write $E^\mathfrak{g}$ for the space of {\it $\mathfrak{g}$-invariant elements} of $E$, namely 
$$
E^\mathfrak{g}=
\{e\in E : \rho_v(e)=0 ,\;\forall v\in \mathfrak{g}\}.
$$
\end{definition}
	
\begin{proposition}\label{Cogmap}
The mappings 
$\Delta^{(m)}:T(\mathfrak{g})\rightarrow T^{(m+1)}(\mathfrak{g})$, 
with $m\in \mathbb{N}$, 
are $\mathfrak{g}$-module morphisms between the natural $\mathfrak{g}$-module structures of $T(\mathfrak{g})$ and $T^{(m+1)}(\mathfrak{g})$, namely
\begin{equation}\label{Proof}
\Delta^{(m)}\circ {\rm ad}_v=
{\rm ad}_v\circ\Delta^{(m)},\qquad 
\forall v\in \mathfrak{g}.
\end{equation}
Moreover, $\Delta^{(m)}({T(\mathfrak{g})}^{\mathfrak{g}})\subset [T^{(m+1)}(\mathfrak{g})]^{\mathfrak{g}}$.
\end{proposition}
\begin{proof} 
Let us prove (\ref{Proof}) by induction. 
In fact, 
${\rm ad}_v\circ \Delta=\Delta\circ {\rm ad}_v$ 
for every $v\in\mathfrak{g}$ and (\ref{Proof}) holds for $m=1$. 
If (\ref{Proof}) is obeyed for a fixed~$m$, then
\vskip-5mm
$$
\Delta^{(m+1)}\circ {\rm ad}_v =
( {\rm Id}\boxtimes
\stackrel{^m}{\ldots}
\boxtimes{\rm Id}
\boxtimes\Delta) \circ 
\Delta^{(m)}\circ {\rm ad}_v =
({\rm Id}\boxtimes
\stackrel{^m}{\ldots}
\boxtimes{\rm Id}\boxtimes\Delta) \circ
{\rm ad}_v\circ \Delta^{(m)} =
{\rm ad}_v\circ \Delta^{(m+1)}
$$
for every $v\in \mathfrak{g}$. 
By induction (\ref{Proof}) is valid for any $m\in \mathbb{N}$ 
and the relation 
$\Delta^{(m)}({T(\mathfrak{g})}^{\mathfrak{g}}) \subset 
[T^{(m+1)}(\mathfrak{g})]^{\mathfrak{g}}$ 
follows trivially.
\end{proof}
	
Let $\mathfrak{T}(N)$ be the unital associative algebra of covariant tensor fields on~$N$.
We define 
$N^m=N\times \stackrel{m}{\ldots} \times N$ 
and let $(\xi_1,\ldots,\xi_m)$ be a point of $N^m$. We write
$\mathfrak{T}^{(m)}(N)=\mathfrak{T}(N)\otimes_{N^m}\stackrel{^m}{\ldots} \otimes_{N^m}\mathfrak{T}(N)$
for the space of covariant tensor fields on $N^m$ being linear combinations of tensor fields of the form
$$
T_1(\xi_1)\otimes_{N^m}\ldots\otimes_{N^{m}}T_m(\xi_m),
$$
where $T_1,\ldots,T_m\in \mathfrak{T}(N)$. 

\begin{theorem} 
\label{th:SupTheo}
 Consider the $\mathfrak{g}$-module structure on $\mathfrak{T}(N)$ 
given by $\rho: \mathfrak{g} \rightarrow {\rm Der}(\mathfrak{T}(N))$.  
If $\iota: \mathfrak{g} \rightarrow \mathfrak{T}(N)$ is a $\mathfrak{g}$-module morphism, 
then $\iota$ can be extended uniquely to a $\mathfrak{g}$-module and associative algebra morphism 
$\Upsilon: T(\mathfrak{g}) \mapsto \mathfrak{T}(N)$
by requiring that
$$
\Upsilon(v_1\otimes_\mathfrak{g} \ldots \otimes_\mathfrak{g} v_r)=\iota(v_1)\otimes_N \ldots \otimes_N \iota(v_r),\qquad \forall v_1,\ldots,v_r\in\mathfrak{g},\qquad \forall r\in \mathbb{N},
$$
where $\otimes_\mathfrak{g}$ is the tensor product in $T(\mathfrak{g})$ whereas $\otimes_N$ is the tensor field product on~$N$. 
	
Moreover, $\rho$ gives rise to a Lie algebra representation $\rho^{(m)}:\mathfrak{g}\rightarrow {\rm Der}(\mathfrak{T}^{(m)}(N))$ such that
$$
\rho^{(m)}_v (T_1\otimes_{N^m}\ldots \otimes_{N^m} T_m)=
\rho_v(T_1) \otimes_{N^m} \ldots \otimes_{N^m} T_m+\ldots+T_1 \otimes_{N^m} \ldots \otimes_{N^m} \rho_v(T_m),
$$
for all $ T_1,\ldots,T_m\in \mathfrak{T}(N)$, every $v\in \mathfrak{g}$, and every $r\in \mathbb{N}$. 
Additionally, there exists a $\mathfrak{g}$-module morphism 
$\Upsilon^{(m)}:T^{(m)}(\mathfrak{g})\rightarrow \mathfrak{T}^{(m)}(N)\subset \mathfrak{T}(N^m)$ such that
$$
\Upsilon^{(m)}({\bf t_1}\boxtimes \stackrel{m}{\ldots} \boxtimes {\bf t_m})=
\Upsilon({\bf t_1})\otimes_{N^m} \ldots \otimes_{N^m} \Upsilon({\bf t_m}),
\qquad \forall {\bf t_1},\ldots,{\bf t_m} \in T(\mathfrak{g}).
$$
\end{theorem}

\begin{figure}[h]\label{Diagram:Th54}
\caption{Diagram summarising the results of Theorem \ref{th:SupTheo}.}
\begin{equation}
\label{StarTrekki1}
\xymatrix{
T(\mathfrak{g}) \ar[d]_{{\rm ad}_v} \ar[rr]^{\Delta^{(m-1)}} &&
T^{(m)}(\mathfrak{g}) \ar[d]_{{\rm ad}_v} \ar[rr]^{\Upsilon^{(m)}} &&
\mathfrak{T}^{(m)}(N) \ar[d]_{\rho^{(m)}_v}
\\
T(\mathfrak{g}) \ar[rr]^{\Delta^{(m-1)}} &&
T^{(m)}(\mathfrak{g}) \ar[rr]^{\Upsilon^{(m)}} &&
\mathfrak{T}^{(m)}(N).
}
\end{equation}
\end{figure}
\begin{proof} 
The map $\Upsilon$ is a well-defined algebra morphism. 
Let us prove that it is a $\mathfrak{g}$-module morphism also. 
Since $\iota$ is a $\mathfrak{g}$-module morphism and $\rho_v$ is a derivation for every $v\in \mathfrak{g}$, 
one obtains that, on decomposable elements of $T(\mathfrak{g})$,
$$
\begin{aligned}
\Upsilon({\rm ad}_v(v_1\otimes_\mathfrak{g}\ldots\otimes_\mathfrak{g} v_r))&=
\Upsilon \left(
\sum_{\alpha=1}^r v_1 \otimes_\mathfrak{g} \ldots \otimes_\mathfrak{g} {\rm ad}_v v_\alpha \otimes_\mathfrak{g} \ldots \otimes_\mathfrak{g} v_r
\right) \\
&= \sum_{\alpha=1}^r \iota(v_1) \otimes_N \ldots \otimes_N \iota ({\rm ad}_v v_\alpha) \otimes_N \ldots \otimes_N \iota(v_r) \\
&= \sum_{\alpha=1}^r \iota(v_1) \otimes_N \ldots \otimes_N \rho_v \iota(v_\alpha) \otimes_N \ldots\otimes_N\iota(v_r) \\
&= \rho_v(\iota(v_1) \otimes_N \ldots \otimes_N \iota(v_r)) \\
&= \rho_v
(\Upsilon (v_1 \otimes_\mathfrak{g} \ldots \otimes_\mathfrak{g} v_r)).
\end{aligned}
$$
As $\Upsilon\circ {\rm ad}_v=\rho_v\circ \Upsilon$ on decomposable elements of $T(\mathfrak{g})$, 
this equality is obeyed on the whole 
$T(\mathfrak{g})$ and 
$\Upsilon\circ {\rm ad}_v=\rho_v\circ \Upsilon$ 
for every $v\in \mathfrak{g}$ and 
$\Upsilon$ becomes a morphism of $\mathfrak{g}$-modules.
	
Let us verify that  $\rho^{(m)}$ is a Lie algebra morphism. As previously, we will start by proving this fact on decomposable elements of $\mathfrak{T}^{(m)}(N)\subset\mathfrak{T}(N^m)$:
$$
\rho_w^{(r)}\rho_v^{(r)}(T_1\otimes_{N^m}\ldots \otimes_{N^m}T_m)=\sum_{i,j=1}^rT_1\otimes_{N^m}\ldots\otimes_{N^m}\rho_w(T_i)\otimes_{N^m}\ldots\otimes_{N^m}\rho_v(T_j)\otimes_{N^m}\ldots \otimes_{N^m}T_m,
$$
for all $v,w\in \mathfrak{g}$ and $T_1,\ldots,T_m\in \mathfrak{T}(N).$ 
Subtracting from this the value of  
$\rho_v^{(r)}\rho_w^{(r)} 
(T_1 \otimes_{N^m} \ldots \otimes_{N^m} T_m)$, 
all elements with $v$ and $w$ acting on elements
$T_i,T_j$ with $i\neq j$ disappear and we are left with
$$
\left[\rho_w^{(r)},\rho_v^{(r)}\right]
(T_1\otimes_{N^m}\ldots \otimes_{N^m}T_m)=
\sum_{i=1}^r
T_1 \otimes_{N^m} \ldots \otimes_{N^m} [\rho_w,\rho_v](T_i) \otimes_{N^m} \ldots \otimes_{N^m} T_m.
$$
Since $\rho$ is a Lie algebra representation, it follows that 
$$
\begin{aligned}
\left[\rho_w^{(r)},\rho_v^{(r)}\right] 
(T_1 \otimes_{N^m} \ldots \otimes_{N^m} T_m)&=
\sum_{i=1}^n T_1 \otimes_{N^m} \ldots \otimes_{N^m} \rho_{[w,v]} (T_i) \otimes_{N^m} \ldots \otimes_{N^m} T_m\\
&= \rho_{[w,v]}^{(r)} (T_1 \otimes_{N^m} \ldots \otimes_{N^m} T_m)
\end{aligned}
$$
on decomposable elements of $\mathfrak{T}^{(m)}(N)$. 
The equality for the whole $\mathfrak{T}^{(m)}(N)$ is therefore also satisfied.
	
The fact that  $\Upsilon^{(m)}$ is a morphism of $\mathfrak{g}$-modules 
results immediately from using the previous ideas.
\end{proof}

\smallskip
It is useful now to consider two particular cases described in the following lemmas. 
\begin{lemma} 
The space $S(\mathfrak{g})$ of totally symmetric tensors over $\mathfrak{g}$ 
and the space $\mathsf{\Lambda}(\mathfrak{g})$ of totally antisymmetric tensors over $\mathfrak{g}$ 
are  $\mathfrak{g}$-submodules of $T(\mathfrak{g})$. 
Similarly, $S^{(m)}(\mathfrak{g})$ and $\mathsf{\Lambda}^{(m)}(\mathfrak{g})$ are  $\mathfrak{g}$-submodules of $T^{(m)}(\mathfrak{g})$.
\end{lemma}
Its proof is an immediate consequence of the fact that if ${\bf t}_1$ is an element totally antisymmetric (symmetric) of $T^{(m)}(\mathfrak{g})$, 
then ${\rm ad}_v({\bf t}_1)$ is totally antisymmetric (symmetric) for every $v\in \mathfrak{g}$.

\smallskip
The following lemma shows that the mappings $\Delta^{(m)}$ can be restricted to symmetric and antisymmetric tensors.
\begin{lemma}
\label{FuckRes} 
For every Lie algebra $\mathfrak{g}$ and $m \in \mathbb{N}$, one has that 
$\Delta^{(m)}S(\mathfrak{g}) \subset S^{(m+1)}(\mathfrak{g})$
and 
$\Delta^{(m)} \mathsf{\Lambda}(\mathfrak{g}) \subset \mathsf{\Lambda}^{(m+1)}(\mathfrak{g})$.
\end{lemma}
The proof of this result is rather technical;
it will be done in the appendix.

\medskip
Lemma \ref{FuckRes} allows us to extend 
Proposition \ref{Cogmap} and
Theorem \ref{th:SupTheo}
to $S(\mathfrak{g})$ and $\mathsf{\Lambda}(\mathfrak{g})$. 
All these results can be summarised through the commutative diagrams displayed in  Figure 2.
\begin{figure}[h]\label{diagram:Dia1}
\caption{Commutative diagrams summarising the results of 	Theorem \ref{th:SupTheo} and Lemma \ref{FuckRes}.
Both diagrams are the restrictions of 
diagram (\ref{StarTrekki1})
to the submodules of symmetric and antisymmetric tensors.
}
\kern -5mm
\begin{equation}
\label{StarTrekki2}
\xymatrix@C=14mm{
\mathsf{S}(\mathfrak{g}) \ar[d]_{{\rm ad}_v} \ar[r]^{\Delta^{(m-1)}} &
\mathsf{S}^{(m)}(\mathfrak{g}) \ar[d]_{{\rm ad}_v} \ar[r]^{\Upsilon^{(m)}} &
\mathfrak{T}^{(m)}(N) \ar[d]_{\rho^{(m)}_v}
\\
\mathsf{S}(\mathfrak{g}) \ar[r]^{\Delta^{(m-1)}} &
\mathsf{S}^{(m)}(\mathfrak{g}) \ar[r]^{\Upsilon^{(m)}} &
\mathfrak{T}^{(m)}(N) 
}
\quad
\xymatrix@C=14mm{
\mathsf{\Lambda}(\mathfrak{g}) \ar[d]_{{\rm ad}_v} \ar[r]^{\Delta^{(m-1)}} &
\mathsf{\Lambda}^{(m)}(\mathfrak{g}) \ar[d]_{{\rm ad}_v} \ar[r]^{\Upsilon^{(m)}} &
\mathfrak{T}^{(m)}(N) \ar[d]_{\rho^{(m)}_v}
\\
\mathsf{\Lambda}(\mathfrak{g}) \ar[r]^{\Delta^{(m-1)}} &
\mathsf{\Lambda}^{(m)}(\mathfrak{g}) \ar[r]^{\Upsilon^{(m)}} &
\mathfrak{T}^{(m)}(N) 
}
\end{equation}
\end{figure}

\subsection{Application to multisymplectic Lie systems: calculus of tensor  invariants}

Let us use the algebraic structures developed in the previous section to derive superposition rules for multisymplectic Lie systems without solving systems of PDEs or ordinary differential equations as in most of the literature \cite{CGM07,Dissertationes,CL11,PW}.

A relevant tool for obtaining a superposition rule is given by diagonal prolongations of vector fields \cite{CGM07}. If $X$ is a vector field on~$N$,
its \emph{diagonal prolongation} to $N^m$
is the vector field
$X^{[m]}(x_{(1)},\ldots,x_{(m)}) =
X(x_{(1)}) + \ldots + X(x_{(m)})$. 
Let us recall the distributional method to obtain superposition rules. 
For further details we refer to 
\cite[Sections 1.5 and 1.6]{Dissertationes}. 
Given a Lie system on $N$ 
with a Vessiot--Guldberg Lie algebra~$V$ 
spanned by a basis of vector fields $X_1,\ldots,X_r$, 
a superposition rule for $X$ can be obtained 
by determining the smallest $m$ so that 
$X_1^{[m]},\ldots,X_r^{[m]}$ 
are linearly independent at a generic point. 
Then, we must obtain $n$ common first-integrals 
$I_1,\ldots,I_n$ for 
$X_1^{[m+1]},\ldots,X_r^{[m+1]}$ 
satisfying that 
\begin{equation}\label{conp}\partial(I_1,\ldots,I_n) / 
\partial(x^1_{(1)},\ldots,x^n_{(1)}) 
\neq 0.
\end{equation}
Let $\lambda_1,\ldots,\lambda_n$ be real numbers. 
By assuming 
$I_1=\lambda_1$, \ldots, $I_n=\lambda_n$, 
condition (\ref{conp}) allows us to express 
$x_{(1)}^1,\ldots,x_{(1)}^n$ 
as functions of $\lambda_1,\ldots,\lambda_n$ 
and the variables $x_{(i)}^1,\ldots,x_{(i)}^n$ 
for $i=2,\ldots,m+1$, 
which gives a superposition rule depending on $m$ particular solutions.

Our algebraic/geometric methods to obtain superposition rules 
rely on obtaining $I_1,\ldots,I_n$ 
through $\mathfrak{g}$-invariant elements 
of the spaces $T^{(q)}(\mathfrak{g})$, 
with $q\in \mathbb{N}$. 
We propose two methods, 
one relying on the Casimir elements 
of the universal enveloping algebra $U(\mathfrak{g})$, 
and another one basing 
on the invariant elements of the Grassmann algebra 
$\mathsf{\Lambda}(\mathfrak{g})$ 
of the linear space~$\mathfrak{g}$.

\begin{proposition}	
If $(N,\Theta,X)$ is a multisymplectic Lie system, 
then $(N^{[m]},\Theta^{[m]},X^{[m]})$ 
is also a multisymplectic Lie system.
\end{proposition}
\begin{proof} 
By assumption, the Vessiot--Guldberg Lie algebra $V^X$ related to $X$ 
consists of Hamiltonian vector fields relative to the multisymplectic form $\Theta$, i.e. 
$\iota_{X_t}\Theta=d\theta_{t}$ 
for certain differential forms $\theta_{t}$, with $t\in \mathbb{R}$. 
Since $\Theta^{[m]}(x_{(1)},\ldots,x_{(m)})=\sum_{a=1}^m\Theta(x_{(a)})$, 
then $\Theta^{[m]}$ is closed and 1-nondegenerate. 
Additionally, 
$$
\iota_{X_t^{[m]}}\Theta^{[m]}=\sum_{a=1}^md\theta_{t}(x_{(a)}),
$$
and the vector fields $X_t^{[m]}$ are Hamiltonian relative to $\Theta^{[m]}$ for every $t\in \mathbb{R}$. 
Hence, $(N^{[m]},\Theta^{[m]},X^{[m]})$ is a multisymplectic Lie system.
\end{proof}

It has been shown in Proposition \ref{prop:LHalgebra}
that 
every multisymplectic Lie system $(N,\Theta,X)$ of degree~$k$
induces on the set 
$\mathfrak{W}$
of the the differentials of the Hamiltonian forms of~$V$ 
a Lie algebra structure,
inherited from the Lie algebra structure on the space of differential $(k{-}1)$-forms
defined in Definition \ref{def:AH}. 
Let us apply the results of the preceding section to~$\mathfrak{W}$.

\begin{proposition}
\label{prop:CorRep} 
Let $(N,\Theta,X)$ be a multisymplectic Lie system with an induced Lie--Hamilton algebra $\mathfrak{W}$. 
Consider a Lie algebra isomorphism $\phi:\mathfrak{g}\simeq \mathfrak{W}$. 
Then, $\mathfrak{T}^{(m)}(N)$ becomes a $\mathfrak{g}$-module with respect to the Lie algebra representation
$$
\rho^{(m)}: v\in \mathfrak{g} \mapsto \mathcal{L}_{-X^{[m]}_v}\in {\rm Der}(\mathfrak{T}^{(m)}(N)),
$$
where $X_v$ is the unique Hamiltonian vector field (relative to $\Theta$) associated with $\phi(v)\in \mathfrak{M}$. 
Moreover, $\phi$ can be extended to a $\mathfrak{g}$-module morphism  
$\Upsilon^{(m)}:T^{(m)}(\mathfrak{g})\rightarrow \mathfrak{T}^{(m)}(N)$, i.e.
\begin{equation}
\label{Eq:RelRel}
\Upsilon^{(m)}({\rm ad}_v({\bf t})) =
\mathcal{L}_{-X^{[m]}_v} \Upsilon^{(m)}({\bf t}),\qquad 
\forall v\in \mathfrak{g},
\forall\, {\bf t}\in T^{(m)}(\mathfrak{g}).
\end{equation}
\end{proposition}

\begin{proof}  
The Lie algebra isomorphism $\phi:\mathfrak{g}\simeq\mathfrak{M}$ allows us to define a linear morphism
\begin{equation}\label{Eq:LH}
\rho: v\in \mathfrak{g} \mapsto \mathcal{L}_{-X_v} \in {\rm Der}(\mathfrak{T}(N)),
\end{equation}
where $X_v$ is the unique Hamiltonian vector field such that $\iota_{X_v}\Theta=\phi(v)$.
Let us show that $\rho$ is a Lie algebra morphism. It is immediate that $\rho$ is linear. From the properties of the Lie bracket (\ref{braform}), 
$\phi([v,\hat{v}])=\{\phi(v),\phi(\hat v)\}$ is the Hamiltonian form of $-[X_v,X_{\hat v}]$. Hence, 
$$
\rho([v,\hat v])=
\mathcal{L}_{X_{[v,\hat v]}}=
\mathcal{L}_{[-X_v,-X_{\hat v}]}=
[\rho_{v},\rho_{\bar v}], \qquad \forall
v,\bar v\in \mathfrak{g},
$$
and $\rho$ is a Lie algebra morphism.
Moreover, $\phi$ can be considered as an injection 
$\phi:\mathfrak{g}\simeq\mathfrak{M}\subset \mathfrak{T}(N)$ of 
$\mathfrak{g}$ in $\mathfrak{T}(N)$ 
and it is also a $\mathfrak{g}$-module morphism. 
Hence, all assumptions of Theorem \ref{th:SupTheo} hold and we can apply it to our particular case. 
More specifically, $\phi$ can be extended to an algebra morphism 
$\Upsilon^{(m)}:T^{(m)}(\mathfrak{g})\rightarrow \mathfrak{T}^{(m)}(N)$.
	
Theorem \ref{th:SupTheo} states that $\rho$ can be extended to a Lie algebra morphism 
$\rho^{(m)}: \mathfrak{g} \rightarrow {\rm Der}(\mathfrak{T}^{(m)}(N))$.  
It is interesting to observe that on decomposable elements of 
$\mathfrak{T}^{(m)}(N)$, 
namely elements of the form 
$T_1\otimes_{N^m}\ldots \otimes_{N^m}T_m$ 
with $T_1,\ldots,T_m\in \mathfrak{T}(N)$, 
one has that
$$
\rho^{(m)}_v(T_1\otimes_{N^m}\ldots \otimes_{N^m}T_m) =
\mathcal{L}_{-X_v^{[m]}}
(T_1\otimes_{N^m}\ldots \otimes_{N^m}T_m),
$$
where $X_v^{[m]}$ is the prolongation to $N^m$ of the Hamiltonian vector field $X_v$ associated with $v\in \mathfrak{g}$. 
As a consequence, one can assume that 
$\rho^{(m)}_v=\mathcal{L}_{-X_v^{[m]}}$. 
Since Theorem \ref{th:SupTheo} ensures that $\Upsilon^{(m)}$ is a morphism of $\mathfrak{g}$-modules, one obtains the relation (\ref{Eq:RelRel}).
\end{proof}

The Poincar\'e--Birkhoff--Witt theorem may be employed to prove that 
every element of the enveloping Lie algebra $U(\mathfrak{g})$ 
can be understood as a unique symmetric element of the tensor algebra 
$T(\mathfrak{g})$ and vice versa 
(cf.\ \cite{Va84}). 
In other words, there exists a linear isomorphism 
$\lambda:U(\mathfrak{g})\rightarrow S(\mathfrak{g})$ 
identifying both linear spaces. 
Moreover, 
$\lambda$ is also a $\mathfrak{g}$-module morphism \cite{Va84} 
and $S(\mathfrak{g})$ is isomorphic (as a $\mathfrak{g}$-module) 
to the symmetric algebra of~$\mathfrak{g}$, 
namely the algebra of commutative polynomials in the elements of $\mathfrak{g}$ 
(see \cite{Va84} for details). 
The elements of $U(\mathfrak{g})$ that commute with any other element of~$\mathfrak{g}$ 
relative to its $\mathfrak{g}$-module structure, 
the so-called {\it Casimir elements}, 
give rise via $\lambda$ to elements of $[S(\mathfrak{g})]^\mathfrak{g}$.
From these comments, 
Lemma \ref{FuckRes} and Proposition \ref{prop:CorRep},
we will obtain the following result:

\begin{corollary}
\label{cor:Inv} 
Let $(N,\Theta,X)$ be a multisymplectic Lie system with a Lie--Hamilton algebra~$\mathfrak{M}$. 
Let 
$\phi:\mathfrak{g}\simeq\mathfrak{M}$ 
be a Lie algebra isomorphism and let $\Upsilon^{(m)}:T^{(m)}(\mathfrak{g})\rightarrow \mathfrak{T}^{(m)}(N)$ 
be its induced morphism of $\mathfrak{g}$-algebras given in 
Proposition \ref{prop:CorRep}. 
If $C$ is a Casimir element of $U(\mathfrak{g})$ 
or an element of $\mathsf{\Lambda}(\mathfrak{g})^\mathfrak{g}$, 
then 
$\Upsilon^{(m)}(\Delta^{(m-1)}C)$ 
is an invariant relative to the evolution of $X^{[m]}$.  
\end{corollary}
\begin{proof} 
Let $V$ be the Lie algebra of Hamiltonian vector fields of~$\mathfrak{W}$. 
If $Y\in V$, then 
$Y^{[m]}\in V^{[m]}$ and $V^{[m]}$ is spanned by the vector fields $Y^{[m]}$, 
where $Y$ is an arbitrary element of~$V$. 
It is immediate that $V^{[m]}$ is a Lie algebra isomorphic to~$V$. 

If $\theta_Y=\iota_{Y}\Theta$, then 
$\theta_{Y^{[m]}}=\iota_{Y^{[m]}}\Theta^{[m]}=\theta_{Y}^{[m]}$.
As a consequence $\mathfrak{W}^{[m]}$ 
is isomorphic to $\mathfrak{W}$ 
relative to the Lie brackets (\ref{braform}) 
induced by $\Theta^{[m]}$ and~$\Theta$, respectively. 
Moreover, $\mathfrak{M}^{[m]}$ becomes a Lie--Hamilton algebra for 
$(N^{[m]},\Theta^{[m]},X^{[m]})$. 
	
Let $v\in \mathfrak{g}$ be such that $\theta_Y=\phi(v)$. 
Thus,  
$\Upsilon^{(m)}(\Delta^{(m-1)}(v))=
\theta_Y^{[m]}=\theta_{Y^{[m]}}$. 
In other words, $\Upsilon^{(m)}(\Delta^{(m-1)}(v))$ is the Hamiltonian form corresponding to $Y^{[m]}$. 
If $C$ is a Casimir of $U(\mathfrak{g})$ or an element of 
$\mathsf{\Lambda}(\mathfrak{g})^\mathfrak{g}$, 
then its symmetric or antisymmetric representative in $T(\mathfrak{g})$ 
is a $\mathfrak{g}$-invariant element of $T(\mathfrak{g})$. 
By using Diagram (\ref{StarTrekki2}), 
we obtain that
\begin{equation}
\mathcal{L}_{Y^{[m]}}[\Upsilon ^{(m)}(\Delta^{(m-1)}{C})]=
\Upsilon^{(m)}\circ\rho_v^{(m)}(\Delta^{(m-1)}(C))=
\Upsilon^{(m)}\circ\Delta^{(m-1)}({\rm ad}_v(C))=0
\end{equation}
for every $Y\in V$.
Hence, 
$\Upsilon^{(m)}(\Delta^{(m-1)}(C))$ 
is an invariant relative to the Vessiot--Guldberg Lie algebra $V^{[m]}$ of $X^{[m]}$ 
and it becomes invariant under the evolution of $X^{[m]}$.
\end{proof}

It is worth noting that $\Upsilon^{(m)}$ is not an algebra morphism relative to the product $\otimes$ in $T^{(m)}(\mathfrak{g})$. 
For instance,
\begin{equation*} 
\Upsilon^{(2)}[(1\boxtimes v_1) \otimes (v_2\boxtimes 1)] =
\Upsilon^{(2)}(v_2\boxtimes v_1) =
\Upsilon(v_2)(x_1)\otimes_{N^2} \Upsilon(v_1)(x_2),
\end{equation*}
whereas 
$$
\Upsilon^{(2)}(1\boxtimes v_1) \otimes_{N^2} \Upsilon^{(2)}(v_2\boxtimes 1) =
\Upsilon(v_1)(x_2)\otimes_{N^2} \Upsilon(v_2)(x_1).
$$
Therefore, both expressions only coincide when
$$
\Upsilon(v_2)(x_1)\otimes_N \Upsilon(v_1)(x_2) = 
\Upsilon(v_1)(x_2)
\otimes_N \Upsilon(v_2)(x_1).
$$

\subsection{Casimir elements and superposition rules}
Let us illustrate in this section how to apply the formalism devised in the previous one to obtain, 
via the Casimir of $\mathfrak{sl}_2$, 
a superposition rule for the multisymplectic Lie system related to Schwarz equations.

Consider the Schwarz equation given by (\ref{Eq:KS3}). 
Its first-order system of of differential equations (\ref{Eq:firstKS3}) is related to a Vessiot--Guldberg Lie algebra 
$V^S=\langle X_1,X_2,X_3\rangle$, 
where $X_1,X_2,X_3$ are given in (\ref{Eq:VFKS1}). 
As shown in Section \ref{subsection:definition_properties_MLS}, 
this Lie system is related to a multisymplectic Lie system 
$(\mathcal{O},\Theta_S,X^S)$ 
with a Lie--Hamilton algebra of differential two-forms 
$\mathfrak{M}_S = \langle d\theta_1,d\theta_2,d\theta_3\rangle$, 
whith $d\theta_1,d\theta_2,d\theta_3$ 
given by (\ref{dfSE}), 
isomorphic to $\mathfrak{sl}_2$ relative to the Lie bracket (\ref{braform}) induced by $\Theta_S$. 

Let $\{v_1,v_2,v_3\}$ be a basis of $\mathfrak{sl}_2$ satisfying the commutation relations
$$
[v_1,v_2]= -v_1,\qquad 
[v_1,v_3]= -2v_2,\qquad 
[v_2,v_3]= -v_3.
$$
The associated universal enveloping algebra $U(\mathfrak{sl}_2)$ has essentially a unique Casimir element 
\cite{BCHLS13,Pressley} 
whose symmetric tensorial form is given by
$$
\mathcal{C}=v_1\otimes v_3+v_3\otimes v_1-2v_2\otimes v_2.
$$
The corresponding morphism of algebras 
$\Upsilon: U(\mathfrak{sl}_2) \rightarrow T(\mathcal{O})$ 
gives rise to an invariant (relative to the Lie derivatives with elements of $V^S$) 4-covariant tensor field on $\mathcal{O}$ given by
$$
\Upsilon(\mathcal{C})=d\theta_1\otimes d\theta_3+d\theta_3\otimes d\theta_1-2d\theta_2\otimes d\theta_2.
$$
It is simple to verify by direct computation that this tensor field is invariant relative to the Lie derivatives with respect to elements of the Vessiot--Guldberg Lie algebra $V^S$. 

The triple $(\mathcal{O},X^S,V^S)$ is an automorphic Lie system. 
Therefore, it satisfies the conditions of Lemma \ref{lemma:Exis}, 
which ensures that $X^S$ admits a Lie algebra ${\rm Sym}(V^S)$ isomorphic to $V^S$ of Lie symmetries of $V^S$ and $X^S$. 
In view of Lemma \ref{lemma:Exis}, the multisymplectic form $\Theta_S$ is also invariant relative to the elements of ${\rm Sym}(V^S)$. 
It is known that $V^S$ admits a Lie algebra of Lie symmetries spanned by \cite{LG99,OT05} the vector fields $Y_1,Y_2,Y_3$ given by (\ref{Eq:SymSc2}).

The contractions of the tensor field $\Upsilon(\mathcal{C})$ with four elements of
${\rm Sym}(V^S)$
are constants of motion for $X^S$. 
In particular, a long but simple calculation shows that 
$$
\iota_{Y_1}\iota_{Y_3}\iota_{Y_1}\iota_{Y_3} \Upsilon(\mathcal{C}) = -2,
\qquad
\iota_{Y_2}\iota_{Y_1}\iota_{Y_2}\iota_{Y_3} \Upsilon(\mathcal{C}) = -1.
$$
To obtain a superposition rule for the system, we have to obtain three functionally independent constants of motion for the diagonal prolongation of $X^S$ to $\mathcal{O}^2$ 
(cf.\ \cite{Dissertationes}). 
Let us write $d\theta^{(j)}_i=d\theta_i(x_{(j)})$. 
Then, the extended invariant $\mathcal{I}^{[2]}=\Upsilon^{(2)}(\Delta (\mathcal{C}))$ reads
$$
\Upsilon(\mathcal{C})^{[2]}+2[d\theta_3(\xi_1)\otimes d\theta_1(\xi_2)+d\theta_1(\xi_1)\otimes d\theta_3(\xi_2)-2d\theta_2(\xi_1)\otimes d\theta_2(\xi_2)]. 
$$
It is worth noting that $\mathcal{I}^{[2]}\neq \Upsilon(\mathcal{C})^{[2]}$.

In virtue of Corollary \ref{cor:Inv}, 
the contractions of the tensor field $\Upsilon(\Delta (\mathcal{C}))$ with $Y_1^{[2]},Y^{[2]}_2,Y^{[2]}_3$
are also invariants of $(X^S)^{[2]}$. 
This fact also allows us to obtain constants of motion for $(X^S)^{[2]}$, 
and, since the vector fields $X_1,X_2,X_3$ are linearly independent at a generic point, 
to obtain superposition rules for~$X^S$ 
in view of the distributional method 
(cf.\ \cite{CGM07,Dissertationes}).

The contractions of $\Upsilon(\mathcal{C})^{[2]}$ with four arbitrary vector fields of $\langle Y_1^{[2]},Y^{[2]}_2,Y^{[2]}_3\rangle$ satisfy that
$$
\iota_{Y^{[2]}_a}\iota_{Y^{[2]}_b}\iota_{Y^{[2]}_c}\iota_{Y^{[2]}_d}\Upsilon(\mathcal{C})^{[2]}=\iota_{Y_a}\iota_{Y_b}\iota_{Y_c}\iota_{Y_d}\Upsilon(\mathcal{C})(x_{(1)})+\iota_{Y_a}\iota_{Y_b}\iota_{Y_c}\iota_{Y_d}\Upsilon(\mathcal{C})(x_{(2)}),
$$
for all $Y_a,Y_b,Y_c,Y_d\in {\rm Sym}(V^S)$. 
As a consequence, such contractions are constants and therefore useless for our purposes as they will lead to trivial constants that cannot be used to obtain a superposition rule. 
Meanwhile, if $C=\mathcal{I}^{[2]}-\Upsilon(\mathcal{C})^{[2]}$, then the contractions 
$$
\begin{gathered}
I_1 =
\frac{(a_2v_1-a_1v_2)^2}{v_1^3v_2^3} =
2\iota_{Y^{[2]}_1}\iota_{Y^{[2]}_2}\iota_{Y^{[2]}_1}\iota_{Y^{[2]}_2} C \,,
\\
I_2 =
\frac{2v_1v_2(v_1-v_2)}{a_2v_1-v_2a_1} +x_1+x_2 =
2\iota_{Y^{[2]}_1}\iota_{Y^{[2]}_2}\iota_{Y^{[2]}_1}\iota_{Y^{[2]}_3} \frac{C}{I_1} \,,
\\ 
I_3 =
\left(x_2-\frac{2v_1v_2^2}{a_2v_1-a_1v_2}\right)
\left(I_2-x_2+\frac{2v_1v_2^2}{a_2v_1-a_1v_2}\right) =
\iota_{Y^{[2]}_1}\iota_{Y^{[2]}_3}\iota_{Y^{[2]}_1}\iota_{Y^{[2]}_3} \frac{C}{2I_1} \,.
\end{gathered}
$$
give rise to three constants of motion for $(X^S)^{[2]}$. 
Remarkably, the contractions of $C$ with elements of ${\rm Sym}(V^S)$ offer many other constants of motion for the diagonal prolongation $(X^S)^{[2]}$, which allows us to select those having a simpler or more appropriate form to obtain a superposition rule. In particular, $I_1,I_2,I_3$ are chosen so as that
\begin{equation}\label{condition}
\partial(I_1,I_2,I_3)/\partial (x_1,v_1,a_1)\neq 0,
\end{equation}
which implies that $I_1,I_2,I_3$ are functionally independent and enable us to obtain a superposition rule 
(cf.\ \cite{CGM07,Dissertationes}). 
Since $X^{[2]}_1,X^{[2]}_2,X^{[2]}_3$ span a distribution or rank three almost everywhere on the six-dimensional manifold $\mathcal{O}^2$, 
all remaining constants of motion for a generic $(X^S)^{[2]}$ are of the form $F=F(I_1,I_2,I_3)$ for a certain function 
$F \colon \mathbb{R}^3 \rightarrow \mathbb{R}$. 

The above procedure is more powerful than the methods devised in \cite{CGLS14,LS13} 
to obtain the constants of motion for $(X^{S})^{[2]}$. 
Indeed, the constants of motion obtained in \cite{LS13} 
were derived via the characteristics method for 
$X^{[2]}_1,X^{[2]}_2,X^{[2]}_3$, which is very tedious. 
Meanwhile, the Dirac structure method employed in \cite{CGLS14} is simpler to be applied 
than the procedure in \cite{LS13} 
but it still demands to obtain a Dirac structure, which may be long, 
along with some Lie symmetries of ${\rm Sym}(V^S)$ 
and some other symmetries chosen in an {\it ad-hoc} way. 
Instead, our present methods give rise to the multisymplectic form in an immediate manner and, knowing several Lie symmetries, one can get the superposition rule. 
Observe that other contractions of $\Upsilon(\mathcal{C})$ or $\Upsilon^{[2]}(\Delta(\mathcal{C}))$ 
with the Lie symmetries (\ref{Eq:SymSc2}) 
and their prolongations to $\mathcal{O}^2$ 
can be used to obtain other tensorial invariants to study the Schwarz equation, 
e.g.\ presymplectic forms given by 
$\iota_{Y_i\wedge Y_i} \Upsilon(\mathcal{C})$ 
with $i,j=1,2,3$.

To obtain the searched superposition rule for the Schwarz equation, 
we use the functions $I_1,I_2,I_3$, which satisfy condition (\ref{condition}). 
Therefore, the following functions
$$
\Upsilon_1 = I_1 \,,
\quad 
\Upsilon_2 = \frac{I_2\pm\sqrt{I_2^2-4I_3}}{2} =
x_2-\frac{2v_1v_2^2}{a_2v_1-a_1v_2} \,, 
\quad
\Upsilon_3 = I_2-\Upsilon_2 = x_1+\frac{2v_1^2v_2}{a_2v_1-a_1v_2} \,, 
$$
also satisfy (\ref{condition}) and the equations 
$\Upsilon_1=k_1$,
$\Upsilon_2=k_2$,
$\Upsilon_3=k_3$ 
can be employed to obtain $x_1,v_1,a_1$ 
from $x_2,v_2,a_2$ and $k_1,k_2,k_3$. 
Indeed, these are the equations employed in \cite{CGLS14} 
to obtain the superposition rule for Schwarz equations. 
Here, they appear without the necessity of using symmetries apart from the Lie symmetries of ${\rm Sym}(V^S)$. 

More specifically, 
the superposition rule can be obtained by using the equation $\Upsilon_1=k_1$, 
which enables us to obtain the value of $a_2v_1-v_2a_1$ in terms of $k_1,v_1,v_2$. 
Substituting this into equations 
$\Upsilon_2=k_2$ and $\Upsilon_3=k_3$, 
we obtain two algebraic equations concerning the variables $v_1,v_2,x_1,x_2$ and the constants $k_1,k_2,k_3$. 
This allows us to obtain $x_1,v_1$ in terms of $x_2,v_2,k_1,k_2,k_3$. 
In particular, one finds that 
\begin{equation}
\label{Homo}
x_1 = \frac{\alpha x_2+\beta}{\gamma x_2+\delta} \,, 
\end{equation}
for certain constants $\alpha,\beta,\gamma,\delta$ 
satisfying that $\alpha\delta-\beta\gamma=1$ 
and whose form can be expressed as a function of $k_1,k_2,k_3$. 
The expression of
$v_1$ and $a_1$ in terms of $x_2,v_2,a_2,\alpha,\beta,\gamma,\delta$ 
can be then easily obtained from (\ref{Homo}). 
The resulting expressions become a superposition rule for the Schwarz equation 
(cf.\ \cite{CGM07,Dissertationes}).

\subsection{Invariant forms and superposition rules}

Let us study an application of the methods of the previous section to derive, 
via invariant elements of the Grassmann algebra $\Lambda(\mathfrak{sl}_2)$, 
a superposition rule for the multisymplectic Lie system related to a control system.

Consider the Riccati-type diffusion system 
\begin{equation}\label{Eq:Partial01}
\left\{
\begin{aligned}
\frac{dw}{dt} &= a(t)v^2,\\
\frac{du}{dt} &= -b(t)+2c(t) u+4a(t)u^2, 
\\
\frac{dv}{dt} &= \left( c(t)+4a(t)u \right) v, 
\\
\end{aligned}
\right.
\end{equation}
where $a(t)$, $b(t),$  and $c(t)$ are arbitrary $t$-dependent functions. 
This system appears as a reduction of a system of differential equations 
that is used to solve diffusion-type equations, Burger's equations, and other PDEs 
\cite{SSV14}.  
Moreover, its relation to Dirac structures has been studied in \cite{CGLS14}. 
Let us apply our methods to obtain its properties.

For simplicity, we restrict ourselves to analysing the system (\ref{Eq:Partial01}) on 
$N= \{(u,v,w)\in\mathbb{R}^3 \mid v\neq 0\}$.  
This highlights the main points of our presentation by avoiding secondary technical details, e.g.\ all hereafter given structures are well-defined on $N$.

The system (\ref{Eq:Partial01}) describes the integral curves of the $t$-dependent vector field
\[
X^{RS}_t = a(t)X_1-b(t)X_2+c(t)X_3\,,
\]
on $N$, where 
\begin{equation}\label{Eq:BD}
\begin{gathered}
X_1 =4u^2\frac{\partial}{\partial u}+ 4uv\frac{\partial}{\partial v}+ v^2\frac{\partial}{\partial w},\qquad
X_2 = 2u\frac{\partial}{\partial u}+v\frac{\partial }{\partial  v},\qquad  
X_3 = \frac{\partial}{\partial u}.
\end{gathered}
\end{equation}
span 
a Lie algebra $V^R$ isomorphic to $\mathfrak{sl}_2$. 
In fact,
\begin{equation}
\label{commutation}
[X_1,X_2]=-2X_1,\qquad 
[X_2,X_3]=-2X_3,\qquad 
[X_1,X_3]=-4X_2. 
\end{equation}
Hence, $V^R$ is a Vessiot--Guldberg Lie algebra for the system (\ref{Eq:Partial}) which becomes a Lie system.

Since $(du\wedge dv\wedge dw)(X_1,X_2,X_3)=v^4$, 
one has that 
$\mathcal{D}^{V^R}_p=T_pN$ 
for any $p\in N$ and $\dim V^R=\dim N$. 
Therefore, $(N,X,V^R)$ is a locally automorphic Lie system.
Since the vector fields $X_1,X_2,X_3$ are linearly independent at a generic point, 
$X$ admits a superposition rule depending on a unique particular solution.
Recall that Lemma \ref{lemma:Exis} ensures that all Lie symmetries for the elements of $V^R$ span a Lie algebra ${\rm Sym}(V^R)$ isomorphic to $V^R$. 
A long but simple calculation allows us to obtain that ${\rm Sym}(V^R)$ is spanned by:
\begin{equation}
\label{SymSch}
\begin{gathered}
Y_1 = v^2\frac{\partial}{\partial u}+4vw\frac{\partial}{\partial v}+4w^2\frac{\partial}{\partial w},\qquad 
Y_2 = v\frac{\partial}{\partial v}+2w\frac{\partial}{\partial w},\qquad
Y_3=\frac{\partial}{\partial w}.
\end{gathered}
\end{equation}

The corresponding dual one-forms to the vector fields in (\ref{Eq:BD}) read
\begin{equation}
\label{DiffOneS}
\begin{gathered}
\eta_1 = \frac{dw}{v^2},\qquad 
\eta_2 = \frac{dv}{v}-4\frac{udw}{v^2},\qquad 
\eta_3 = du-\frac{2udv}{v}+\frac{4u^2dw}{v^2}.
\end{gathered}
\end{equation}
Since $V^R$ is semisimple, it is therefore unimodular and to obtain a multisymplectic form invariant under the Lie derivatives with elements of $V^R$, 
one has just to  
define the multisymplectic form 
$$
\Theta^{RS}=\eta_1\wedge \eta_2\wedge \eta_3=\frac{1}{v^3}dw\wedge dv\wedge du
$$
turning $(N,\Theta^{RS},X^{RS})$ into a multisymplectic Lie system.  

The differentials forms $\iota_{X_\alpha}\Theta^{RS}$, with $\alpha=1,2,3$, 
span a Lie--Hamilton algebra $\mathfrak{M}$ isomorphic to~$V^R$.
In view of the bracket (\ref{braform}) and (\ref{commutation}), their commutation relations are
\begin{equation}
\label{Eq:comm2}
\begin{gathered}{} 
[d\theta_1,d\theta_2] = 2d\theta_1, 
\qquad 
[d\theta_1,d\theta_3] = 4d\theta_2,
\qquad 
[d\theta_2,d\theta_3] = 2d\theta_3. 
\end{gathered}
\end{equation}

To obtain invariants and superposition rules for $X^{RS}$, 
we will derive invariants of the Grassmann algebra $\mathsf{\Lambda}(\mathfrak{sl}_2)$. 
More exactly, we will use an element of 
$\mathsf{\Lambda}(\mathfrak{sl}_2)^{\mathfrak{sl}_2}$, 
e.g.\ 
$\mathcal{C}=v_1\wedge v_2\wedge v_3$, 
and then $\Upsilon^{(2)}\Delta(\mathcal{C})$ will be invariant under the evolution of $(X^{RS})^{[2]}$ in view of Corollary \ref{cor:Inv}.

More specifically, 
\begin{multline*}
\Delta(v_1\wedge v_2\wedge v_3) =
(v_1\wedge v_2\wedge v_3)\boxtimes1 + 1\boxtimes (v_1\wedge v_2\wedge v_3) + (v_1\wedge v_2)\boxtimes v_3+ (v_2\wedge v_3)\boxtimes v_1
\\
+(v_3\wedge v_1)\boxtimes v_2 + v_3\boxtimes(v_1\wedge v_2) + v_2\boxtimes (v_3\wedge v_1) + v_1\boxtimes( v_2\wedge v_3).
\end{multline*}
It is worth noting that, as proved in Lemma \ref{FuckRes}, we have that $\Delta(v_1\wedge v_2\wedge v_3)\subset \mathsf{\Lambda}(\mathfrak{sl}_2)\otimes \mathsf{\Lambda}(\mathfrak{sl}_2)$.

A simple calculation gives that 
$$
\iota_{Y_1}\iota_{Y_2}\iota_{Y_1}\iota_{Y_3}\iota_{Y_2}\iota_{Y_3} 
\Upsilon^{(2)}\Delta(v_1\wedge v_2\wedge v_3) = 
-2 \, \frac{[v_1^2+v_2^2-4(u_1-u_2)(w_1-w_2)]^2}{v_1^2v_2^2}.
$$
This can be simplified to
$$
f_1: =\frac{v_1^2+v_2^2-4(u_1-u_2)(w_1-w_2)}{v_1v_2}
$$

The application of the Lie symmetries $Y_1,Y_2,Y_3$ to $f_1$ 
gives rise to the following two functionally independent invariants:
$$
f_2=\frac{u_2-u_1}{v_1v_2} \,,
\qquad 
f_3=\frac{v_1^2-v_2^2-4(u_1-u_2)(w_1+w_2)}{v_1v_2} \,.
$$
Since $\partial(f_1,f_2,f_3)/\partial(v_1,u_1,w_1)\neq 0$,
by equating $f_1,f_2,f_3$ to constants $k_1,k_2,k_3$, 
it is possible to obtain $u_1,v_1,w_1$ 
as functions of $u_2,v_2,w_2$ and the constants $k_1,k_2,k_3$ 
and to give rise to a superposition rule for the system under study.

As commented before, the above system admits a superposition rule depending on just one particular solution. 
Hence, one has to obtain constants of motion for the diagonal prolongation of the system to $N^{[2]}$ 
to obtain the superposition rule.

\section{Conclusions and outlook}

This work has illustrated the existence of multisymplectic Lie systems in the literature 
and has provided tools to endow Lie systems 
with a compatible multisymplectic structure. 
This has lead us to endow them with certain algebraic structures that enable the obtention of superposition rules and invariants,
retrieving as particular cases much of the invariants 
found in the previous literature on the topic. 
This seems to have applications to extend the coalgebra method in the literature 
\cite{BBHMR09},
which would have applications in the theory of integrable systems. 
Results have been illustrated by examples of physical and mathematical interest.

In the future, 
we aim to extend the approach given in multisymplectic Lie systems to arbitrary Lie systems 
by attaching to them an associative tensor algebra 
obtained by tensor products of the elements 
of a Vessiot--Guldberg Lie algebra 
and, eventually, their dual forms. 
The Vessiot--Guldberg Lie algebra acts then on this associative algebra and 
its invariants should give rise to invariant structures for Lie systems. 
We believe that this approach will recover all results of the whole literature on geometric structures on Lie systems as particular cases 
\cite{BBHLS13,BCHLS13,CGLS14,GL17,LL18,LV15}. 
Finally, it seems that the ideas of this work 
could be applied to generalise 
not only the coalgebra formalism for obtaining superposition rules for Lie--Hamilton systems 
but also the coalgebra method itself 
(see \cite{BBHMR09}). 
Additionally, we believe that 
the quantum algebra structure of the universal enveloping algebra 
appearing in the coalgebra method 
could be employed to obtain new integrals of motion.

\appendix
\section*{Appendix: proof of Lemma \ref{FuckRes}}
\addcontentsline{toc}{section}{Appendix}
\def\thesection{A}

\paragraph{Lemma}
{\itshape\!\!
Let $\mathfrak{g}$ be a Lie algebra and $m \in \mathbb{N}$. Then,  
$\Delta^{(m)}S(\mathfrak{g}) \subset S^{(m+1)}(\mathfrak{g})$
and 
$\Delta^{(m)} \mathsf{\Lambda}(\mathfrak{g}) \subset \mathsf{\Lambda}^{(m+1)}(\mathfrak{g})$.
}

\begin{proof} 
Let's prove the inclusion for $S(\mathfrak{g})$.
It will be enough to prove it for homogeneous elements of~$S(\mathfrak{g})$.
We first consider the case $m=1$. 
	
Consider elements $v_1,\ldots,v_r$ of $\mathfrak{g}$. 
Let us write 
$v_{i_1<\ldots<i_k} = v_{i_1} \otimes \ldots \otimes v_{i_k}$, 
where $1\leq i_1<\ldots<i_k\leq r$, 
the $r$ is any natural number, 
and $v^c_{i_1<\ldots<i_k}$ is
the exterior product of the elements in 
$\{v_{i_1},\ldots,v_{i_k}\}^c$, where $^c$ stands for the complementary set in $\{v_1,\ldots,v_r\}$, 
ordered with respect to their indexes from the lower to the higher one.  
Using the fact that $\Delta$ is a morphism of associative algebras, we get that 
$$
\Delta(v_{i_1<\ldots<i_k})=
\Delta(v_{i_1})\cdot\ldots\cdot \Delta(v_{i_r})=
(v_1\boxtimes 1+1\boxtimes v_1) \cdot \ldots \cdot (v_r\boxtimes 1+1\boxtimes v_r).
$$
It follows by induction that the product on the right-hand side can be written as
$$
\sum_{k=0}^r
\sum_{1\leq i_1<\ldots<i_k\leq r} 
v_{i_1<\ldots<{i_k}} \boxtimes v^c_{i_1<\ldots<{i_k}}.
$$
If we write 
${\rm Alt}(v_1\otimes\ldots\otimes v_r)=
\sum_{\sigma \in S_r}\sigma(v_1\otimes\ldots\otimes v_r)$, 
where $S_r$ is the permutation group of $r$ elements, then 
\begin{equation}\label{1}
\Delta({\rm Alt}(v_{1}\otimes\ldots\otimes v_{r}))=
\sum_{\sigma\in S_r} \sigma \left[
\sum_{k=0}^r \sum_{1\leq i_1<\ldots<i_k\leq r} 
\left( v_{i_1<\ldots<i_k} \boxtimes v^c_{i_1<\ldots<i_k} \right)
\right] \,.
\end{equation}
To prove that the right-hand side belongs to 
$S(\mathfrak{g})\boxtimes S(\mathfrak{g})$, 
we aim to decompose the elements of $\sigma\in S_r$ 
in a specific way so as to control how they act on the elements 
$v_{i_1<\ldots<i_k} \boxtimes v^c_{i_1<\ldots<i_k}$.

Let $G_{i_1<\ldots<i_k}$ be the subgroup of $S_r$ whose elements leave the subset 
$\{i_1,\ldots,i_k\}$ 
invariant and let $G_{i_1,\ldots,i_k}$ be the
subgroup of $G_{i_1<\ldots<i_k}$ whose elements fix all the elements $v_{i_1},\ldots,v_{i_k}$. 
Let $\hat{\sigma}_1,\ldots,\hat \sigma_{s_1}$ and 
$\tilde{\sigma}_1,\ldots,\tilde {\sigma}_{s_2}$ 
be representatives of the equivalence classes of the cosets $S_r/G_{i_1<\ldots<i_k}$ and 
$G_{i_1<\ldots<i_k}/G_{i_1,\ldots,i_k}$, respectively. 
Then, 
$S_r=
\bigcup_{\alpha=1}^{s_1}
\bigcup_{\beta=1}^{s_2}
\bigcup_{\sigma'\in G_{i_1,\ldots,i_k}}
\hat\sigma_\alpha \circ {\tilde{\sigma}}_{\beta}\circ\sigma'$. 
Every $\sigma'\in G_{i_1,\ldots,i_k}$ fixes all the elements of $\{i_1,\ldots,i_k\}$ and leaves invariant the subset $\{i_1,\ldots,i_k\}^c$. 
Hence, 
$\sigma'(v_{i_1<\ldots<i_k})=v_{i_1<\ldots<i_k}$ 
whereas 
$\sum_{\sigma'\in G_{i_1,\ldots,i_k}} \sigma'(v^c_{i_1<\ldots<i_k})=
{\rm Alt} (v^c_{i_1<\ldots<i_k})$.  
Using the previous decomposition of $S_r$ and  (\ref{1}), 
we obtain
\begin{equation}\label{2}
\Delta({\rm Alt}(v_{1}\otimes\ldots\otimes v_{r}))=
\sum_{k=0}^{r}
\sum_{1\leq i_1<\ldots<i_k\leq r}
\sum_{\alpha=1}^{s_1}
\sum_{\beta=1}^{s_2}
\hat\sigma_\alpha \circ \tilde\sigma_\beta[ 
(v_{i_1<\ldots<i_k}) \hat\otimes 
{\rm Alt}(v^c_{i_1<\ldots<i_k})]
\end{equation}
The element $\bar \sigma_\beta$ leaves stable $\{v_{i_1},\ldots,v_{i_r}\}$ and its complementary set. 
Since ${\rm Alt}(v_{i_1<\ldots<i_r}^c)$ is totally symmetric, one has that 
$\bar \sigma_\beta\left[{\rm Alt}(v_{i_1<\ldots<i_r}^c)\right]=
{\rm Alt}(v_{i_1<\ldots<i_r}^c)$. 
Since $G_{i_1<\ldots<i_k}/G_{i_1,\ldots,i_k}\simeq S_k$, 
one has that 
$$
\sum_{\beta=1}^{s_2} \bar \sigma_\beta[v_{i_1<\ldots<i_k}]=
{\rm Alt}(v_{i_1<\ldots<i_k}).
$$
Using previous results in (\ref{2}), one gets
$$\Delta({\rm Alt}(v_{1}\otimes\ldots\otimes v_{r})) =
\sum_{\alpha=1}^{s_1} 
\hat\sigma_\alpha[
{\rm Alt}(v_{i_1<\ldots<i_k}) \boxtimes 
{\rm Alt}(v^c_{i_1<\ldots<i_k}) ]
\in S^{(2)}(\mathfrak{g})
$$
and 
$\Delta S(\mathfrak{g})\subset S^{(2)}(\mathfrak{g})$. 

Once the case for $m=1$ has been proved, 
the case for general~$m$ can be proved by induction. 
Assume that the result is true for $m$ and therefore $\Delta^{(m)}(S(\mathfrak{g})) \subset S^{(m+1)}(\mathfrak{g})$. 
Let us prove that our lemma holds true for $m+1$. 
By the recurrence relation for $\Delta^{(m+1)}$ and the induction hypothesis, 
one has that
$$
\Delta^{(m+1)}(S(\mathfrak{g}))=
({\rm Id}\boxtimes \stackrel{^m}{\ldots}\boxtimes {\rm Id}\boxtimes \Delta)\circ \Delta^{(m)} (S(\mathfrak{g}))=
({\rm Id}\boxtimes \stackrel{^m}{\ldots}\boxtimes {\rm Id}\boxtimes \Delta)(S^{(m)}(\mathfrak{g})) \subset 
S^{(m+1)}(\mathfrak{g}).
$$

The proof for 
$\mathsf{\Lambda}(\mathfrak{g})$ 
follows analogously by defining 
$\overline{\rm Alt}(v_{i_1,\ldots,i_r})=
\sum_{\sigma\in S_r} (-1)^{{\rm sign}(\sigma)}  \sigma(v_{i_1,\ldots,i_r})$ 
and using the following equalities:
$$
\sum_{\sigma'\in G_{i_1,\ldots,i_k}} 
(-1)^{{\rm  sign}(\sigma')} \sigma'(v^c_{i_1<\ldots<i_k}) =
\overline{\rm Alt} (v^c_{i_1<\ldots<i_k})
\,,
$$
$$
\sum_{\beta=1}^{s_2} 
(-1)^{{\rm sign}(\bar{\sigma}_\beta)} 
\bar{\sigma}_\beta
[(v_{i_1<\ldots<i_k}) \boxtimes \overline{\rm Alt}(v^c_{i_1<\ldots<i_k})] =
\sum_{\beta=1}^{s_2} 
\overline{\rm Alt} (v_{i_1<\ldots<i_k}) 
\boxtimes 
\overline{\rm Alt} (v^c_{i_1<\ldots<i_k})
\,.
$$
\end{proof}

\subsection*{Acknowledgements}

The authors acknowledge fruitful discussions on the topic of the paper with our colleague N. Rom\'an-Roy.
We acknowledge partial financial support from 
the Polish National Science Centre project 
2016/22/M/ST1/00542 (HARMONIA);
from the Spanish Ministerio de Econom\'{\i}a y Competitividad projects 
MTM2014--54855--P
and
MTM2015--64166--C2-1-P;
from the Catalan Government project 
2017--SGR--932;
and from the Aragon Government grant E38$\_$17R.


\end{document}